%% file: main.tex
\documentclass[journal]{IEEEtran}

\input{include/0-requirement}

\begin{document}

\title{DiffRMST: A GPU-Accelerated Fully Differentiable Rectilinear Minimum Spanning Tree Wirelength Model for Global Placement}

\author{Fuxing Huang, Hao Ai, Wenxing Zhu, Xinning Liu, and Ziran Zhu%
\thanks{This work has been submitted to the IEEE for possible publication. Copyright may be transferred without notice, after which this version may no longer be accessible.}
\thanks{Preliminary version of this paper was presented at the 2026 63rd ACM/IEEE Design Automation Conference (DAC), Long Beach, CA, USA, July 2026~\cite{DAC26}.}
\thanks{This work was partially supported by the National Natural Science Foundation of China under Grants 62674046 and 92464202.
(Corresponding author: Ziran Zhu).}
\thanks{Fuxing Huang, Xinning Liu, and Ziran Zhu are with the School of Integrated Circuits, Southeast University, Nanjing, China (emails: \{fxhuang, xinning.liu, zrzhu\}@seu.edu.cn).}
\thanks{Hao Ai and Wenxing Zhu are with the Center for Discrete Mathematics and Theoretical Computer Science, Fuzhou University, Fuzhou, China (email: \{2600310002, wxzhu\}@fzu.edu.cn).}
}

\markboth{}{}

\maketitle
\begin{abstract}
Analytical global placement requires a differentiable wirelength objective that accurately reflects routed interconnect cost.
The widely used half-perimeter wirelength (HPWL) is efficient but overlooks the internal topology of multi-pin nets. The rectilinear Steiner minimum tree (RSMT) provides higher fidelity, but its underlying decision problem is NP-complete, and its discrete construction precludes direct use in gradient-based optimization. The rectilinear minimum spanning tree (RMST) offers a polynomial-time alternative, but remains nondifferentiable because of minimization over discrete spanning trees.
To address these limitations, we present DiffRMST, a fully differentiable RMST wirelength model for analytical global placement.
DiffRMST combines spatial smoothing with a log-sum-exp relaxation to define a continuous distribution over spanning trees.
Based on the Kirchhoff matrix-tree theorem, the partition function and edge marginals can be computed without explicit tree enumeration, enabling a closed-form topology-aware gradient for placement optimization. For efficient integration into analytical placement, we further develop a fused forward-backward operator with parallel CPU and GPU implementations.
On the ISPD2019 benchmark suite, our method reduces RMST, RSMT, and routed wirelength by an average of 5.0\%, 3.2\%, and 3.5\%, respectively, compared with the weighted-average baseline, while maintaining acceptable runtime. Furthermore, the GPU implementation achieves an average $14.7\times$ speedup in global placement over its eight-thread CPU counterpart.
\end{abstract}


\input{include/1-introduction}
\input{include/2-preliminaries}

\input{include/3-algorithm}
\input{include/4-experiment}
\input{include/5-conclusion}

\bibliographystyle{IEEEtran}
\bibliography{ref}

\vspace{-10mm}
\begin{IEEEbiography}[{\biophoto{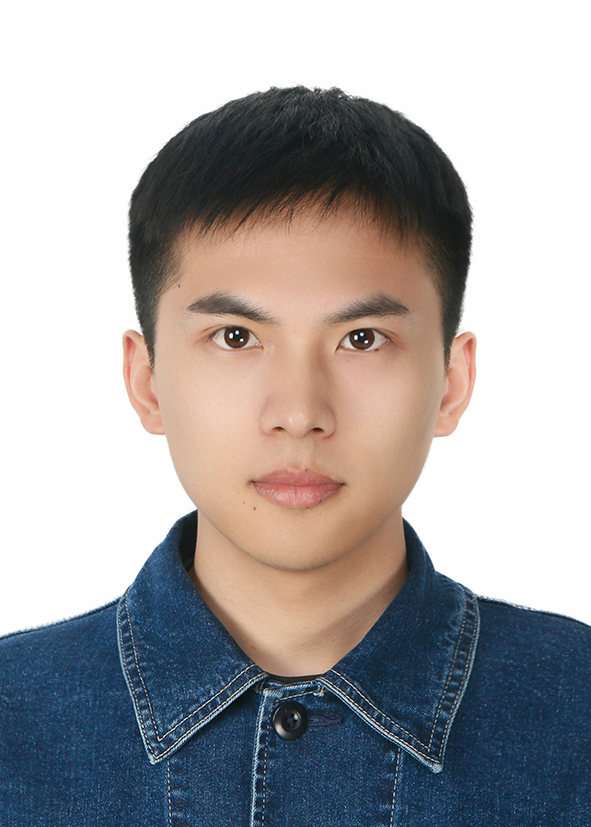}}]
{Fuxing Huang} received his M.Sc. degree in computer science from Fuzhou University, Fuzhou, China, in 2024.
He is currently pursuing the Ph.D. degree at the School of Integrated Circuits, Southeast University, Nanjing, China.
His research interests include mathematical optimization techniques and GPU acceleration for EDA physical design automation.
\end{IEEEbiography}

\vspace{-10mm}
\begin{IEEEbiography}[{\biophoto{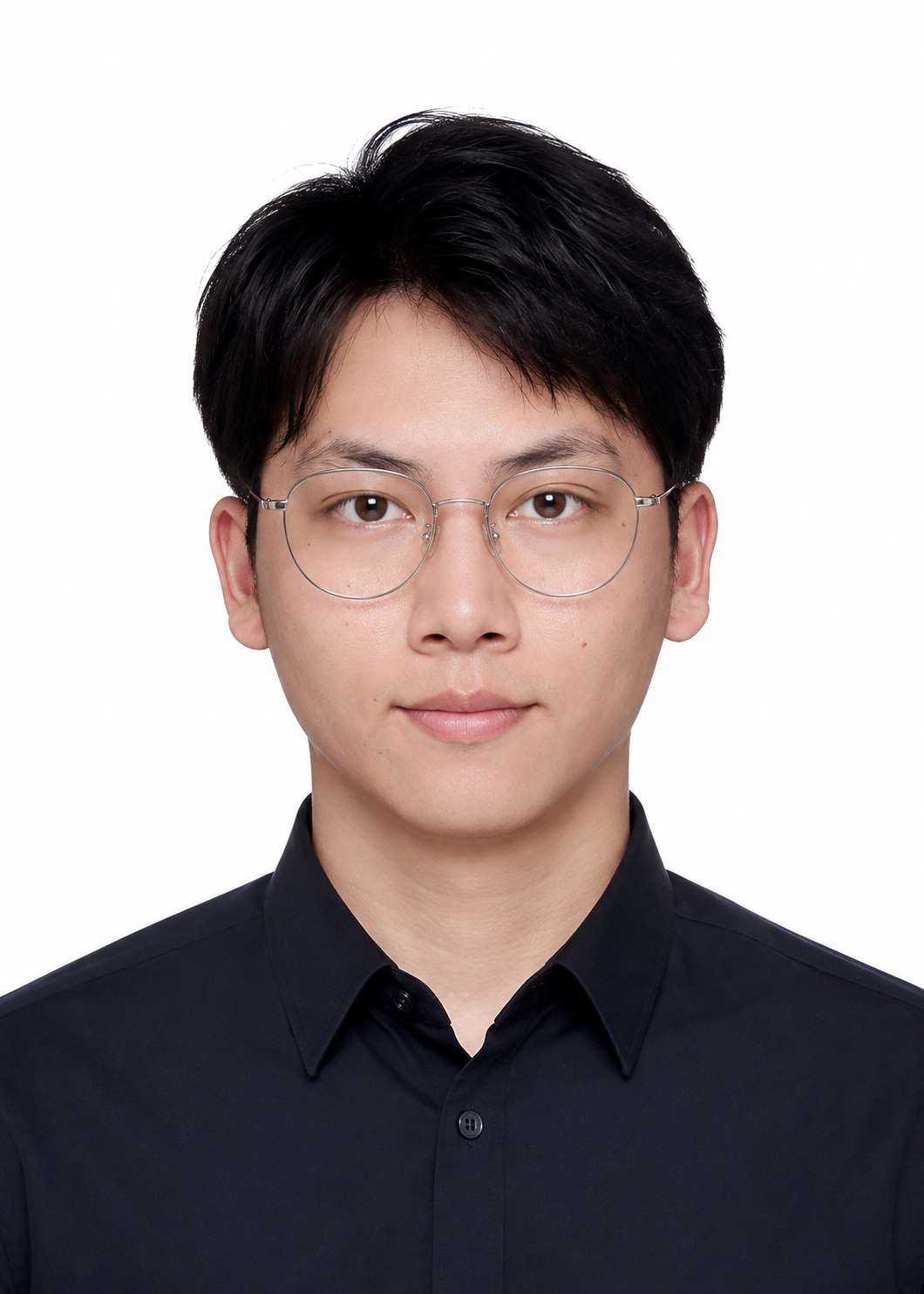}}]
{Hao Ai} is currently pursuing Ph.D. degree at the Center for Discrete Mathematics and Theoretical Computer Science, Fuzhou University, Fuzhou, China. His research interests include VLSI physical design automation and multiphysics simulation for integrated circuits. He is also a recipient of
the Best Paper Award Nomination in DAC 2026.
\end{IEEEbiography}

\vspace{-10mm}
\begin{IEEEbiography}[{\biophoto{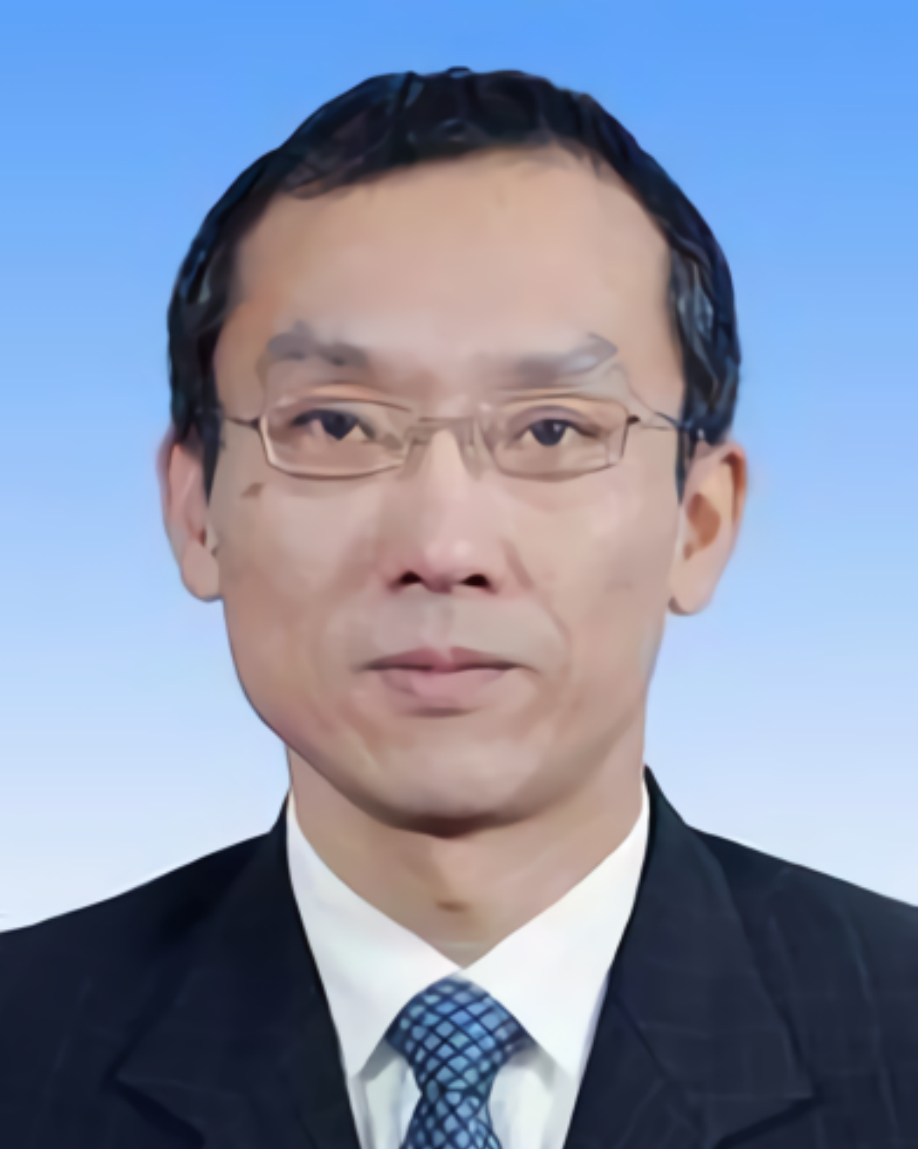}}]
{Wenxing Zhu} received the Ph.D. degree in operations research from Shanghai University, Shanghai, China, in 1996. He is a ``Jiaxi'' Distinguished Professor at Fuzhou University, Fuzhou, China. His research interests include VLSI physical design automation, optimization theory, and algorithms. Dr. Zhu received several prestigious awards, including the First Prize of Natural Sciences from the Ministry of Education of China in 2022 and the Operations Research Application Award from the Operations Research Society of China in 2020. He received the Best Paper Awards from DAC 2017 and ISEDA 2023 and was nominated for the Best Paper Award at DAC 2026, ICCAD 2018, and CCF-DAC 2023. In addition, he was honored with the Second Prize for National Teaching Achievements in 2009.
\end{IEEEbiography}

\vspace{-10mm}
\begin{IEEEbiography}[{\biophoto{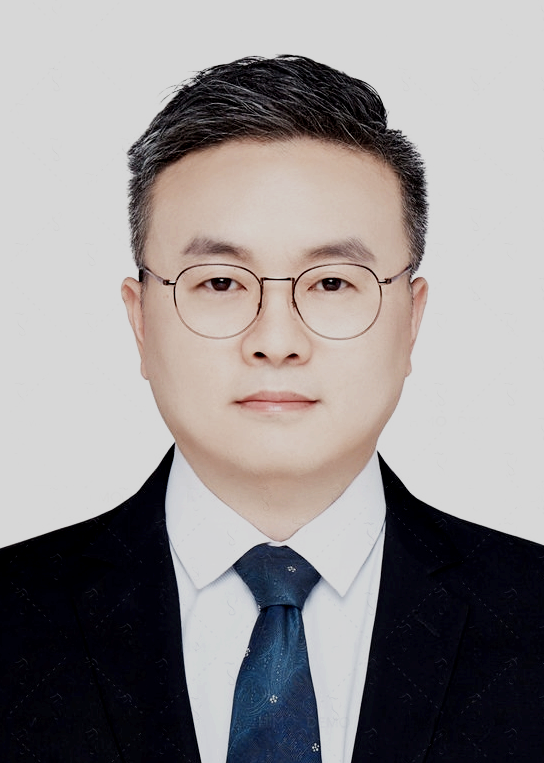}}]
{Xinning Liu} received the B.S., M.S., and Ph.D. degrees from Southeast University, Nanjing, China, in 2000, 2003, and 2015, respectively. He is currently an Associate Professor with the School of Integrated Circuits, Southeast University. His research interests include low power circuits, energy harvesting circuits, and low power SoC design technologies.
\end{IEEEbiography}

\vspace{-10mm}
\begin{IEEEbiography}[{\biophoto{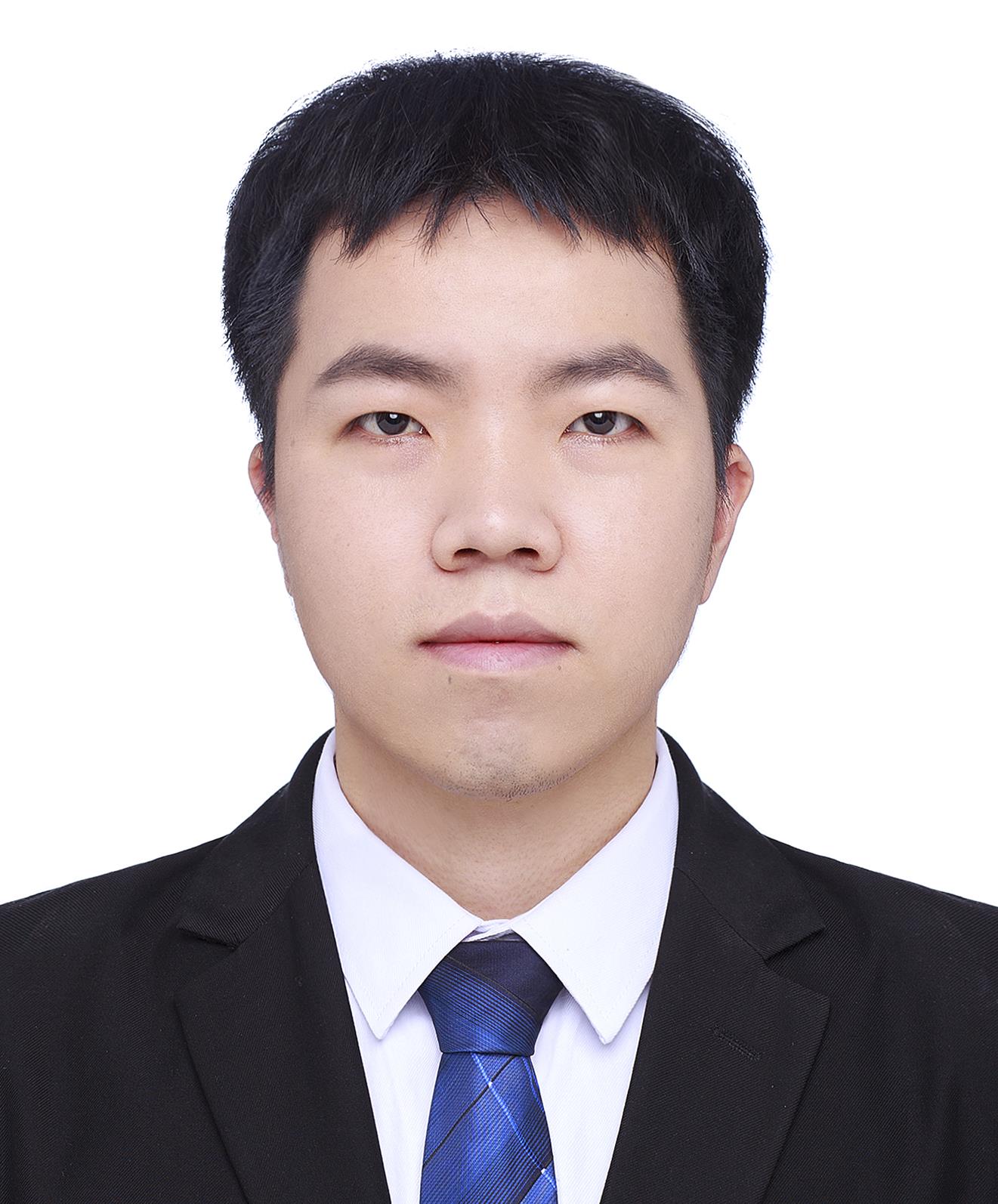}}]
{Ziran Zhu} received the B.Sc. degree in computer science and technology, and the Ph.D. degree in applied mathematics from Fuzhou University, in 2015 and 2020, respectively. He is currently an Associate Professor with Southeast University.
His research interest is EDA physical design automation.
He received the Best Paper Award in DAC 2017, the First Place in CAD Contest @ICCAD 2017 and 2018, the First Place in MLCAD Contest 2023, and the Second Place in CAD Contest @ICCAD 2022.
\end{IEEEbiography}
\end{document}

%% file: include/0-requirement.tex
\usepackage{amsmath,amsfonts,amssymb}
\usepackage{graphicx}
\usepackage{tikz}
\usepackage{algorithm}
\usepackage{algpseudocode}
\usepackage{booktabs}
\usepackage{multirow}
\usepackage{cite}
\usepackage[colorlinks=true,linkcolor=blue,citecolor=blue,urlcolor=blue]{hyperref}
\usepackage[normalem]{ulem}

\graphicspath{{include/Figure/}}

\newsavebox{\biophotobox}
\newcommand{\biophoto}[1]{%
    \sbox{\biophotobox}{\includegraphics[width=1in]{#1}}%
    \ifdim\ht\biophotobox<1.25in
        \sbox{\biophotobox}{\includegraphics[height=1.25in]{#1}}%
    \fi
    \begin{tikzpicture}[baseline=0pt]
        \path[use as bounding box] (0,0) rectangle (1in,1.25in);
        \clip (0,0) rectangle (1in,1.25in);
        \node[overlay,inner sep=0pt,outer sep=0pt] at (0.5in,0.625in)
            {\usebox{\biophotobox}};
    \end{tikzpicture}%
}

\allowdisplaybreaks[2]

\newtheorem{theorem}{Theorem}
\newtheorem{corollary}{Corollary}

%% file: include/1-introduction.tex
\section{Introduction}\label{sec:intro}

As VLSI technology continues to scale, circuit placement has become a decisive stage in the physical design flow. It determines the physical locations of cells and macros and establishes the geometric foundation for routing, clock tree synthesis, timing closure, and subsequent physical optimization. Its quality directly affects the performance, power, and area (PPA) of the final chip, as well as routing congestion and thermal reliability \cite{lin2019dreamplace,shi2025timing,fu2024hybrid,liu2023xplace,hsu2014ntuplace4h,liu2025theplace,liu2013fast}. 
Among the objectives considered in global placement, wirelength is particularly important because it is closely related to routing demand, interconnect delay, switching power, and the difficulty of timing and congestion optimization.

Modern analytical global placers typically formulate and solve large-scale nonlinear optimization problems with gradient-based methods, such as Nesterov's accelerated gradient method \cite{sun2019big,lu2015eplace,zhu2018analytical,lin2019dreamplace,liu2023xplace,chen2008ntuplace3,chan2006mpl6,viswanathan2007fastplace,kim2011simpl,kahng2006faster,agnesina2023autodmp}. Against this backdrop, the wirelength model is not merely an evaluation metric. Its gradient determines the attractive force applied to cells at every global placement iteration and directly affects the trajectory and convergence of the placer. The wirelength model must therefore be smooth and differentiable with respect to cell coordinates, while remaining sufficiently faithful to the routed interconnect. 
Because the objective and its gradient are evaluated across all nets at every iteration, their computational cost directly affects the overall placement runtime.
A smooth but inaccurate model can efficiently optimize the wrong objective, whereas an accurate but discrete or expensive model is difficult to use directly at each global placement iteration.
Thus, balancing modeling accuracy, differentiability, and computational efficiency is essential for developing effective wirelength objectives for analytical global placement.

\input{include/Figure/routing_type}

\subsection{Previous Work}

Wirelength estimation models in physical design can be broadly divided into the half-perimeter wirelength (HPWL), rectilinear minimum spanning tree (RMST), and rectilinear Steiner minimum tree (RSMT) models, as illustrated in Figure~\ref{fig:routing_type}. HPWL describes a net by the half-perimeter of its bounding box and is designed for efficient optimization. RMST connects all pins with a rectilinear spanning tree without introducing Steiner points, whereas RSMT further allows Steiner points and yields a minimum-length rectilinear Steiner tree for the given pin locations.
These models form a natural spectrum. HPWL has the lowest computational cost but provides the least topological information, RSMT has the highest routing fidelity but is difficult to solve and optimize, and RMST lies between them by retaining an explicit tree structure without the full combinatorial complexity of Steiner-point construction. The following paragraphs review the limitations of existing work along this spectrum.

HPWL has been the most widely used wirelength objective in analytical placement because its bounding-box expression is inexpensive and its nonsmooth maximum and minimum operators admit differentiable approximations. The weighted-average (WA) and log-sum-exp (LSE) models are widely used to smooth HPWL for gradient-based optimization \cite{naylor2001non,hsu2011tsv,hsu2013tsv}. Moreau-envelope, $L_p$-type, and CHKS models further study the approximation error, gradient support, convexity, and numerical behavior of smooth HPWL formulations \cite{liao2023moreau,kennings2000analytical,chen1993non}. The force-field analysis in \cite{liao2023moreau} shows that conventional dense smoothing can assign unnecessary forces to pins that are internal to the current bounding box, which can oppose density-driven spreading during overlap removal. The Moreau-envelope model can produce sparse gradients and deactivate such internal pins, but it still optimizes the same bounding-box surrogate. For a net with more than two pins, HPWL does not reveal which pins should be connected, where internal branches should pass, or how much wire is required inside the bounding box. Consequently, existing HPWL models improve optimization smoothness but cannot remove the fidelity gap between the placement objective and the actual routed wirelength.

RSMT methods directly construct rectilinear routing trees with Steiner points and are therefore widely used for routing, interconnect planning, and post-placement wirelength or timing estimation. The Hanan-grid property shows that an optimal RSMT can be restricted to a finite set of candidate Steiner locations \cite{hanan1966steiner}. However, the RSMT problem is NP-complete \cite{garey1977rsmt}, and exact optimization is not available in polynomial time unless P equals NP. Approximation schemes such as Arora's PTAS provide strong theoretical guarantees, but their practical complexity is too high for repeated calls in a large-scale placement optimizer \cite{arora1998ptas}. Practical methods therefore rely on specialized heuristics. FLUTE uses lookup tables for small nets and decomposes larger nets into smaller subproblems \cite{chu2004flute,chu2005flute,chu2007flute}. GeoSteiner targets exact solutions for relatively small instances \cite{juhl2018geosteiner}. Recent learning-based methods, including REST and NN-Steiner, accelerate the construction of discrete Steiner trees while preserving the routing objective \cite{liu2021rest,kahng2024nn}. These approaches are effective for routing and reference estimation, but their discrete tree selection does not provide a smooth end-to-end gradient with respect to cell coordinates. For example, DREAMPlace 4.0 uses FLUTE-derived trees as an RC-tree reference or for timing-related net weighting, while its placement wirelength objective remains HPWL \cite{liao2022dreamplace}. Because FLUTE or RSMT construction is much more expensive than evaluating a smooth HPWL surrogate, such tools invoke it periodically rather than at every global placement iteration. Thus, high-fidelity RSMT information remains separated from the continuous optimization trajectory of global placement.

RMST provides an attractive trade-off between the efficiency of HPWL and the fidelity of RSMT. Since its candidate vertices are restricted to the pins of a net, the exact RMST can be obtained in polynomial time. Geometry-based algorithms achieve $\mathcal{O}(n \log n)$ time, while a direct $\mathcal{O}(n^2)$ implementation of Prim's algorithm is often preferred for nets with small degree \cite{guibas1983northeast,cormen1990algorithms,kahng2004rmstpack}. Unlike HPWL, RMST produces an explicit pin-to-pin rectilinear tree and can therefore provide a concrete routing-length estimate for a multi-pin net. Its approximation quality is also controlled by Hwang's theorem, which bounds the RMST length by 1.5 times the optimal RSMT length \cite{hwang1976steiner}. Nevertheless, exact RMST remains a discrete combinatorial objective because Manhattan distances contain absolute-value operations and the minimum over feasible spanning trees changes with the pin coordinates. Thus, exact RMST is nondifferentiable despite being polynomial-time solvable. Existing work has not developed a fully differentiable RMST wirelength model for analytical global placement, leaving this HPWL-RSMT trade-off unavailable to the continuous optimizer.

\subsection{Our Contributions}
Existing wirelength models leave a gap between routing fidelity and suitability for continuous optimization: HPWL is efficient but overlooks net topology, while RSMT offers higher fidelity but is computationally expensive and discrete. To address this gap, we propose DiffRMST, a fully differentiable RMST wirelength model for analytical global placement. It makes RMST's balance between routing fidelity and computational efficiency available to gradient-based optimization. The main contributions are summarized as follows:

\begin{itemize}
\item To the best of our knowledge, this work is the first to introduce a fully differentiable RMST wirelength objective into analytical global placement, enabling its direct optimization by gradient descent. It captures the internal pin-to-pin connectivity omitted by HPWL without the discrete Steiner-point optimization required by RSMT.
\item We construct the objective through spatial smoothing and an LSE relaxation of the discrete spanning-tree minimum. The Kirchhoff matrix-tree theorem (MTT)~\cite{koo2007structured} evaluates the partition function and edge marginals without tree enumeration, yielding a closed-form topology-aware coordinate gradient. Approximation bounds and gradient regularity analysis support the formulation.
\item We develop a fused forward and backward operator that reuses one marginal solve to compute the forward cost proxy and the relaxed objective's gradient. Dynamic scheduling across nets on the CPU and degree-aware subgroup and block mappings on the GPU exploit parallelism across and within nets, while workspace reuse limits temporary storage.
\item On the ISPD2019 benchmark suite, DiffRMST reduces RMST, RSMT, and routed wirelength by 5.0\%, 3.2\%, and 3.5\% on average, respectively, relative to WA. It achieves these gains with acceptable runtime overhead and an average global-placement speedup of $14\times$ over its eight-thread CPU implementation.
\end{itemize}

The remainder of this paper is organized as follows. Section~\ref{sec:preliminaries} introduces the placement problem and wirelength models. Section~\ref{sec:framework} presents the differentiable RMST formulation and its analysis. Section~\ref{sec:exp} reports the experiments, and Section~\ref{sec:conclusion} concludes the paper.

%% file: include/Figure/routing_type.tex
\begin{figure}
\centering

\tikzset{every picture/.style={line width=0.75pt}} %

\begin{tikzpicture}[x=0.75pt,y=0.75pt,yscale=-0.85, xscale=0.85]

\draw  [color={rgb, 255:red, 0; green, 0; blue, 0 }  ,draw opacity=0.8 ][line width=1.21]  (0,23) .. controls (0,10.3) and (10.3,0) .. (23,0) -- (92,0) .. controls (104.7,0) and (115,10.3) .. (115,23) -- (115,127) .. controls (115,139.7) and (104.7,150) .. (92,150) -- (23,150) .. controls (10.3,150) and (0,139.7) .. (0,127) -- cycle ;
\draw  [color={rgb, 255:red, 0; green, 0; blue, 0 }  ,draw opacity=0.8 ][fill={rgb, 255:red, 183; green, 209; blue, 226 }  ,fill opacity=1 ] (4.67,77.67) -- (24.67,77.67) -- (24.67,109) -- (4.67,109) -- cycle ;
\draw  [color={rgb, 255:red, 0; green, 0; blue, 0 }  ,draw opacity=0.8 ][fill={rgb, 255:red, 183; green, 209; blue, 226 }  ,fill opacity=1 ] (47.17,47.67) -- (73.67,47.67) -- (73.67,78.17) -- (47.17,78.17) -- cycle ;
\draw  [color={rgb, 255:red, 0; green, 0; blue, 0 }  ,draw opacity=0.8 ][fill={rgb, 255:red, 183; green, 209; blue, 226 }  ,fill opacity=1 ] (89.67,19.67) -- (109.67,19.67) -- (109.67,49.67) -- (89.67,49.67) -- cycle ;
\draw  [color={rgb, 255:red, 0; green, 0; blue, 0 }  ,draw opacity=0.8 ][fill={rgb, 255:red, 183; green, 209; blue, 226 }  ,fill opacity=1 ] (56,109) -- (86,109) -- (86,129) -- (56,129) -- cycle ;
\draw  [color={rgb, 255:red, 56; green, 128; blue, 177 }  ,draw opacity=1 ][fill={rgb, 255:red, 56; green, 128; blue, 177 }  ,fill opacity=1 ] (25,91.17) .. controls (25,89.79) and (26.12,88.67) .. (27.5,88.67) .. controls (28.88,88.67) and (30,89.79) .. (30,91.17) .. controls (30,92.55) and (28.88,93.67) .. (27.5,93.67) .. controls (26.12,93.67) and (25,92.55) .. (25,91.17) -- cycle ;
\draw  [color={rgb, 255:red, 56; green, 128; blue, 177 }  ,draw opacity=1 ][fill={rgb, 255:red, 56; green, 128; blue, 177 }  ,fill opacity=1 ] (63,81.17) .. controls (63,79.79) and (64.12,78.67) .. (65.5,78.67) .. controls (66.88,78.67) and (68,79.79) .. (68,81.17) .. controls (68,82.55) and (66.88,83.67) .. (65.5,83.67) .. controls (64.12,83.67) and (63,82.55) .. (63,81.17) -- cycle ;
\draw  [color={rgb, 255:red, 56; green, 128; blue, 177 }  ,draw opacity=1 ][fill={rgb, 255:red, 56; green, 128; blue, 177 }  ,fill opacity=1 ] (73,106.17) .. controls (73,104.79) and (74.12,103.67) .. (75.5,103.67) .. controls (76.88,103.67) and (78,104.79) .. (78,106.17) .. controls (78,107.55) and (76.88,108.67) .. (75.5,108.67) .. controls (74.12,108.67) and (73,107.55) .. (73,106.17) -- cycle ;
\draw  [color={rgb, 255:red, 56; green, 128; blue, 177 }  ,draw opacity=1 ][fill={rgb, 255:red, 56; green, 128; blue, 177 }  ,fill opacity=1 ] (84,34.5) .. controls (84,33.12) and (85.12,32) .. (86.5,32) .. controls (87.88,32) and (89,33.12) .. (89,34.5) .. controls (89,35.88) and (87.88,37) .. (86.5,37) .. controls (85.12,37) and (84,35.88) .. (84,34.5) -- cycle ;
\draw  [color={rgb, 255:red, 56; green, 128; blue, 177 }  ,draw opacity=0.8 ][line width=1.21]  (86.5,37) -- (86.5,106.17) -- (78,106.17) ;
\draw  [color={rgb, 255:red, 56; green, 128; blue, 177 }  ,draw opacity=0.8 ][line width=1.21]  (27.27,88.42) -- (27.27,34.74) -- (84.23,34.74) ;
\draw  [color={rgb, 255:red, 56; green, 128; blue, 177 }  ,draw opacity=0.8 ][line width=1.21]  (73,106.17) -- (27.5,106.17) -- (27.5,93.67) ;
\draw  [color={rgb, 255:red, 0; green, 0; blue, 0 }  ,draw opacity=0.8 ][line width=1.21]  (135,23) .. controls (135,10.3) and (145.3,0) .. (158,0) -- (227,0) .. controls (239.7,0) and (250,10.3) .. (250,23) -- (250,127) .. controls (250,139.7) and (239.7,150) .. (227,150) -- (158,150) .. controls (145.3,150) and (135,139.7) .. (135,127) -- cycle ;
\draw  [color={rgb, 255:red, 0; green, 0; blue, 0 }  ,draw opacity=0.8 ][fill={rgb, 255:red, 183; green, 209; blue, 226 }  ,fill opacity=1 ] (139.67,77.67) -- (159.67,77.67) -- (159.67,109) -- (139.67,109) -- cycle ;
\draw  [color={rgb, 255:red, 0; green, 0; blue, 0 }  ,draw opacity=0.8 ][fill={rgb, 255:red, 183; green, 209; blue, 226 }  ,fill opacity=1 ] (182.17,47.67) -- (208.67,47.67) -- (208.67,78.17) -- (182.17,78.17) -- cycle ;
\draw  [color={rgb, 255:red, 0; green, 0; blue, 0 }  ,draw opacity=0.8 ][fill={rgb, 255:red, 183; green, 209; blue, 226 }  ,fill opacity=1 ] (224.67,19.67) -- (244.67,19.67) -- (244.67,49.67) -- (224.67,49.67) -- cycle ;
\draw  [color={rgb, 255:red, 0; green, 0; blue, 0 }  ,draw opacity=0.8 ][fill={rgb, 255:red, 183; green, 209; blue, 226 }  ,fill opacity=1 ] (191,109) -- (221,109) -- (221,129) -- (191,129) -- cycle ;
\draw  [color={rgb, 255:red, 56; green, 128; blue, 177 }  ,draw opacity=1 ][fill={rgb, 255:red, 56; green, 128; blue, 177 }  ,fill opacity=1 ] (160,91.17) .. controls (160,89.79) and (161.12,88.67) .. (162.5,88.67) .. controls (163.88,88.67) and (165,89.79) .. (165,91.17) .. controls (165,92.55) and (163.88,93.67) .. (162.5,93.67) .. controls (161.12,93.67) and (160,92.55) .. (160,91.17) -- cycle ;
\draw  [color={rgb, 255:red, 56; green, 128; blue, 177 }  ,draw opacity=1 ][fill={rgb, 255:red, 56; green, 128; blue, 177 }  ,fill opacity=1 ] (198,81.17) .. controls (198,79.79) and (199.12,78.67) .. (200.5,78.67) .. controls (201.88,78.67) and (203,79.79) .. (203,81.17) .. controls (203,82.55) and (201.88,83.67) .. (200.5,83.67) .. controls (199.12,83.67) and (198,82.55) .. (198,81.17) -- cycle ;
\draw  [color={rgb, 255:red, 56; green, 128; blue, 177 }  ,draw opacity=1 ][fill={rgb, 255:red, 56; green, 128; blue, 177 }  ,fill opacity=1 ] (208,106.17) .. controls (208,104.79) and (209.12,103.67) .. (210.5,103.67) .. controls (211.88,103.67) and (213,104.79) .. (213,106.17) .. controls (213,107.55) and (211.88,108.67) .. (210.5,108.67) .. controls (209.12,108.67) and (208,107.55) .. (208,106.17) -- cycle ;
\draw  [color={rgb, 255:red, 56; green, 128; blue, 177 }  ,draw opacity=1 ][fill={rgb, 255:red, 56; green, 128; blue, 177 }  ,fill opacity=1 ] (219,34.5) .. controls (219,33.12) and (220.12,32) .. (221.5,32) .. controls (222.88,32) and (224,33.12) .. (224,34.5) .. controls (224,35.88) and (222.88,37) .. (221.5,37) .. controls (220.12,37) and (219,35.88) .. (219,34.5) -- cycle ;
\draw  [color={rgb, 255:red, 56; green, 128; blue, 177 }  ,draw opacity=0.8 ][line width=1.21]  (221.5,37) -- (221.5,81.17) -- (203,81.17) ;
\draw  [color={rgb, 255:red, 56; green, 128; blue, 177 }  ,draw opacity=0.8 ][line width=1.21]  (162.5,88.67) -- (162.5,81.17) -- (198,81.17) ;
\draw  [color={rgb, 255:red, 56; green, 128; blue, 177 }  ,draw opacity=0.8 ][line width=1.21]  (208,106.17) -- (200.5,106.17) -- (200.5,83.67) ;
\draw  [color={rgb, 255:red, 0; green, 0; blue, 0 }  ,draw opacity=0.8 ][line width=1.21]  (270,23) .. controls (270,10.3) and (280.3,0) .. (293,0) -- (362,0) .. controls (374.7,0) and (385,10.3) .. (385,23) -- (385,127) .. controls (385,139.7) and (374.7,150) .. (362,150) -- (293,150) .. controls (280.3,150) and (270,139.7) .. (270,127) -- cycle ;
\draw  [color={rgb, 255:red, 0; green, 0; blue, 0 }  ,draw opacity=0.8 ][fill={rgb, 255:red, 183; green, 209; blue, 226 }  ,fill opacity=1 ] (274.67,77.67) -- (294.67,77.67) -- (294.67,109) -- (274.67,109) -- cycle ;
\draw  [color={rgb, 255:red, 0; green, 0; blue, 0 }  ,draw opacity=0.8 ][fill={rgb, 255:red, 183; green, 209; blue, 226 }  ,fill opacity=1 ] (317.17,47.67) -- (343.67,47.67) -- (343.67,78.17) -- (317.17,78.17) -- cycle ;
\draw  [color={rgb, 255:red, 0; green, 0; blue, 0 }  ,draw opacity=0.8 ][fill={rgb, 255:red, 183; green, 209; blue, 226 }  ,fill opacity=1 ] (359.67,19.67) -- (379.67,19.67) -- (379.67,49.67) -- (359.67,49.67) -- cycle ;
\draw  [color={rgb, 255:red, 0; green, 0; blue, 0 }  ,draw opacity=0.8 ][fill={rgb, 255:red, 183; green, 209; blue, 226 }  ,fill opacity=1 ] (326,109) -- (356,109) -- (356,129) -- (326,129) -- cycle ;
\draw  [color={rgb, 255:red, 56; green, 128; blue, 177 }  ,draw opacity=1 ][fill={rgb, 255:red, 56; green, 128; blue, 177 }  ,fill opacity=1 ] (295,91.17) .. controls (295,89.79) and (296.12,88.67) .. (297.5,88.67) .. controls (298.88,88.67) and (300,89.79) .. (300,91.17) .. controls (300,92.55) and (298.88,93.67) .. (297.5,93.67) .. controls (296.12,93.67) and (295,92.55) .. (295,91.17) -- cycle ;
\draw  [color={rgb, 255:red, 56; green, 128; blue, 177 }  ,draw opacity=1 ][fill={rgb, 255:red, 56; green, 128; blue, 177 }  ,fill opacity=1 ] (333,81.17) .. controls (333,79.79) and (334.12,78.67) .. (335.5,78.67) .. controls (336.88,78.67) and (338,79.79) .. (338,81.17) .. controls (338,82.55) and (336.88,83.67) .. (335.5,83.67) .. controls (334.12,83.67) and (333,82.55) .. (333,81.17) -- cycle ;
\draw  [color={rgb, 255:red, 56; green, 128; blue, 177 }  ,draw opacity=1 ][fill={rgb, 255:red, 56; green, 128; blue, 177 }  ,fill opacity=1 ] (343,106.17) .. controls (343,104.79) and (344.12,103.67) .. (345.5,103.67) .. controls (346.88,103.67) and (348,104.79) .. (348,106.17) .. controls (348,107.55) and (346.88,108.67) .. (345.5,108.67) .. controls (344.12,108.67) and (343,107.55) .. (343,106.17) -- cycle ;
\draw  [color={rgb, 255:red, 56; green, 128; blue, 177 }  ,draw opacity=1 ][fill={rgb, 255:red, 56; green, 128; blue, 177 }  ,fill opacity=1 ] (354,34.5) .. controls (354,33.12) and (355.12,32) .. (356.5,32) .. controls (357.88,32) and (359,33.12) .. (359,34.5) .. controls (359,35.88) and (357.88,37) .. (356.5,37) .. controls (355.12,37) and (354,35.88) .. (354,34.5) -- cycle ;
\draw  [color={rgb, 255:red, 56; green, 128; blue, 177 }  ,draw opacity=0.8 ][line width=1.21]  (356.5,37) -- (356.5,81.17) -- (338,81.17) ;
\draw  [color={rgb, 255:red, 56; green, 128; blue, 177 }  ,draw opacity=0.8 ][line width=1.21]  (297.5,88.67) -- (297.5,81.17) -- (333,81.17) ;
\draw [color={rgb, 255:red, 56; green, 128; blue, 177 }  ,draw opacity=0.8 ][line width=1.21]    (345.73,81.18) -- (345.5,103.67) ;

\draw (48,152) node [anchor=north west][inner sep=0.75pt]   [align=left] {\textbf{(a)}};
\draw (184,151.5) node [anchor=north west][inner sep=0.75pt]   [align=left] {\textbf{(b)}};
\draw (320,152) node [anchor=north west][inner sep=0.75pt]   [align=left] {\textbf{(c)}};

\end{tikzpicture}
\caption{Wirelength estimation models for a net. (a) HPWL model; (b) RMST model; (c) RSMT model.}
\vspace{-5mm}
\label{fig:routing_type}
\end{figure}
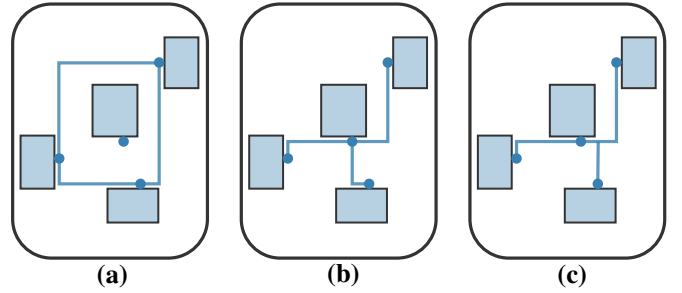

%% file: include/2-preliminaries.tex
\section{Preliminaries}\label{sec:preliminaries}

\subsection{Global Placement Optimization}
In the physical design flow, the primary objective of global placement is to determine the continuous coordinates of cells across the layout area to minimize the total wirelength, subject to cell density constraints. Analytical placers commonly formulate this as an unconstrained nonlinear optimization problem~\cite{lu2015eplace,lin2019dreamplace}:
\begin{equation}
\label{eq:placement_obj}
\min_{x,y} \sum_{e \in E} \mathcal{W}(e; x,y) + \lambda \mathcal{D}(x,y),
\end{equation}
where $E$ denotes the set of all nets in the circuit, and $\mathcal{W}(e;\cdot)$ represents the estimated wirelength of a specific net $e$ based on the current cell locations $(x,y)$. The term $\mathcal{D}(x, y)$ is a density penalty function that penalizes regions with excessive cell congestion, and $\lambda > 0$ is a weighting parameter used to balance the trade-off between wirelength and density. For simplicity in the mathematical model, pin offsets relative to their cell centers are typically ignored, allowing the variables $(x, y)$ to directly represent the coordinates of the pins.

Because modern analytical placement frameworks heavily rely on gradient-based optimization techniques, the objective function must be mathematically smooth. Consequently, any wirelength estimation model $\mathcal{W}(e; \cdot)$ integrated into this framework must be continuously differentiable with respect to the cell coordinates to effectively compute gradients during the optimization process.

\subsection{Wirelength Estimation Models}\label{sec:wirelength_models}
In the objective function of global placement, the wirelength $\mathcal{W}(e; \cdot)$ can be evaluated using different mathematical models.  For a given net $e$, let $V_e$ denote its pin set, with $n=|V_e|$, and let $\mathbf{p}_i=(x_i,y_i)$ denote the position of pin $i\in V_e$.

\subsubsection{Half-Perimeter Wirelength (HPWL)} 
HPWL is the most widely used wirelength estimation model~\cite{kennings2000analytical}, which computes the half-perimeter of the bounding box enclosing all pins in $V_e$:
\begin{equation}
    \mathcal{W}_{\text{HPWL}}(e) = \left( \max_{i \in V_e} x_i - \min_{i \in V_e} x_i \right) + \left( \max_{i \in V_e} y_i - \min_{i \in V_e} y_i \right).
\end{equation}

Because the $\max$ and $\min$ operators are non-differentiable, smooth approximations such as the Weighted-Average (WA)~\cite{hsu2013tsv} or Log-Sum-Exp (LSE)~\cite{naylor2001non} functions are commonly adopted in analytical placers to evaluate HPWL.

\subsubsection{Rectilinear Steiner Tree (RST)} 
To accurately represent the actual routing topology, RST constructs a routing tree that connects all pins in $V_e$ along with a set of optionally inserted Steiner points $S_e$, to achieve the minimum possible routing length. The exact RST wirelength is formulated as:
\begin{equation}
    \mathcal{W}_{\text{RST}}(e) = \min_{S_e, T=(V_e \cup S_e, E_T)} \sum_{\{i,j\} \in E_T} d_{ij},
\end{equation}
where $E_T$ represents the set of edges in the routing tree $T$, and $d_{ij} = |x_i - x_j| + |y_i - y_j|$ is the Manhattan distance between nodes $i$ and $j$. However, finding the minimum RST is a well-known NP-hard problem~\cite{garey1977rsmt}.

\subsubsection{Rectilinear Minimum Spanning Tree (RMST)} 
RMST restricts the tree nodes strictly to the given pin set $V_e$ without allowing any Steiner points. Let $\mathcal{T}_e$ denote the set of all feasible spanning trees that connect all pins in net $e$. The exact RMST wirelength is defined as the following combinatorial optimization problem:
\begin{equation}
    \mathcal{W}_{\text{RMST}}(e; x, y) = \min_{T \in \mathcal{T}_e} \sum_{\{i,j\} \in T} d_{ij}.
    \label{eq:rmst}
\end{equation}
According to Hwang's theorem, the length of an RMST is at most 1.5 times the optimal RST length~\cite{hwang1976steiner}. However, the exact RMST formulation in Equation (\ref{eq:rmst}) is inherently non-differentiable due to two mathematical barriers: the absolute value function in the Manhattan distance $d_{ij}$, and the discrete combinatorial $\min$ operator over the spanning tree set $\mathcal{T}_e$.

\subsection{Mathematical Foundations for Smooth Approximation}

As discussed in the previous section, the exact RMST formulation is non-differentiable. To seamlessly integrate RMST into a gradient-based placement framework, smooth approximation techniques and graph theory foundations are required. 

\subsubsection{Log-Sum-Exp Function} The Log-Sum-Exp (LSE) function is a widely used smooth mathematical approximation for the extreme value functions ($\max$ or $\min$). Given real variables 
$\boldsymbol{z}=(z_1,\ldots,z_m)^{\mathsf T}$, the LSE function for smooth minimum approximation is defined as~\cite{boyd2004convex}:
\begin{equation}
    \label{eq:lse_def}
    \text{LSE}_{\beta}(\boldsymbol{z}) = -\beta \log \sum_{i=1}^m \exp\left(-\frac{z_i}{\beta}\right),
\end{equation}
where $\beta > 0$ is a strictly positive smoothing parameter. The LSE function is continuously differentiable with respect to each $z_i$. Furthermore, it provides a lower bound for the minimum function and uniformly converges to the exact minimum as $\beta \rightarrow 0^+$:
\begin{equation}
    \lim_{\beta \rightarrow 0^+} \text{LSE}_{\beta}(\boldsymbol{z}) = \min_{i=1}^m z_i.
\end{equation}

\subsubsection{Graph Laplacian Matrix}\label{sec:laplacian} To mathematically model the routing topology of a net $e$, we represent it as a complete undirected weighted graph $G_e=(V_e,E_{\text{full}})$, where $E_{\text{full}}$ contains all unordered pairs of distinct pins in $V_e$. Let $\boldsymbol{A}_e=[a_{ij}]$ denote a symmetric matrix of nonnegative edge affinities with $a_{ii}=0$.

The structural and algebraic properties of $G_e$ can be characterized by its Laplacian matrix. The degree matrix and weighted Laplacian $\boldsymbol{L}_e\in\mathbb{R}^{n\times n}$ are defined as
\begin{equation}
    \boldsymbol{D}_e=\operatorname{Diag}(\boldsymbol{A}_e\boldsymbol{1}),
    \qquad\boldsymbol{L}_e=\boldsymbol{D}_e-\boldsymbol{A}_e.
    \label{eq:laplacian_matrices}
\end{equation}
Here, $\operatorname{Diag}$ forms a diagonal matrix from a vector, and $\boldsymbol{1}$ is the all-ones column vector. The degree matrix $\boldsymbol{D}_e$ is distinct from the scalar density penalty $\mathcal{D}$. By definition, the sum of each row and each column in $\boldsymbol{L}_e$ equals zero. The Laplacian matrix formulation is the algebraic foundation for evaluating the sum of products of edge affinities over all spanning trees~\cite{chaiken1982allminors}, which will be utilized in our differentiable RMST model.

%% file: include/3-algorithm.tex
\input{include/Figure/forward}
\section{Proposed DiffRMST Model}\label{sec:framework}
Figure~\ref{fig:forward} summarizes the computation of DiffRMST from pin coordinates to global placement updates. Starting from each net's pin coordinates, we construct a complete graph with smoothed edge weights $\boldsymbol{W}_e$ and convert them into edge affinities $\boldsymbol{A}_e$ under the LSE relaxation. We then form the graph Laplacian $\boldsymbol{L}_e$ and its minor $\boldsymbol{L}_e^k$. Using the weighted matrix-tree theorem, the zero-padded inverse $\boldsymbol{M}_e$ and the edge affinities yield the edge marginal probabilities $\boldsymbol{P}_e$ without enumerating trees. The fused forward and backward computation combines these probabilities with the edge weights to evaluate the expected wirelength as a forward cost proxy, and with the spatial derivatives of the edge weights to obtain coordinate gradients of the relaxed RMST objective. These gradients guide gradient descent updates of the cell coordinates at each global placement iteration. The following subsections develop each component of this formulation and derive the corresponding coordinate gradients.

\subsection{Smooth Approximation and Continuous Relaxation of RMST}
The exact RMST objective in Equation~\eqref{eq:rmst} contains two distinct sources of nondifferentiability. The first arises from the Manhattan edge length, whose absolute-value terms are nondifferentiable when two pins have identical $x$- or $y$-coordinates. The second arises from the discrete minimization over spanning trees, since the optimal tree topology can change abruptly as the pin coordinates vary. We address these two sources separately. 
Spatial smoothing is first applied to obtain differentiable edge lengths, after which a continuous relaxation is introduced over the set of spanning trees to replace the discrete minimum. The parameters $\gamma$ and $\beta$ control edge-length smoothing and tree-topology relaxation, respectively.

\subsubsection{Complete Graph with Smoothed Edge Weights}
The complete graph $G_e$ defined in Section~\ref{sec:laplacian} contains every candidate connection between the pins of net $e$. No tree is selected at this stage. Instead, we assign a smooth geometric length to each pair of pins, preserving the common set of candidate edges throughout coordinate optimization. For a scalar displacement $d\in\mathbb{R}$, the spatial smoothing function is
\begin{equation}
    \phi_\gamma(d)=d\tanh\!\left(\frac{d}{\gamma}\right),
    \label{eq:spatial_smoothing}
\end{equation}
where $\gamma>0$. This function is even and nonnegative, and it satisfies $\phi_\gamma(0)=0$. Around zero, the hyperbolic tangent removes the corner of the absolute value function, providing a continuous derivative when pin coordinates coincide. For a fixed nonzero displacement, $\tanh(d/\gamma)$ approaches the sign of $d$ as $\gamma\to0^+$.

Applying the same smoothing independently to the two coordinate directions gives
\begin{equation}
    w_{ij}=\phi_\gamma(x_i-x_j)+\phi_\gamma(y_i-y_j).
    \label{eq:edge_weight}
\end{equation}
The corresponding edge weight matrix is
\begin{equation}
    \boldsymbol{W}_e=[w_{ij}]
      =\begin{bmatrix}
        0&w_{12}&\cdots&w_{1n}\\
        w_{21}&0&\cdots&w_{2n}\\
        \vdots&\vdots&\ddots&\vdots\\
        w_{n1}&w_{n2}&\cdots&0
      \end{bmatrix}.
    \label{eq:edge_weight_matrix}
\end{equation}
Because $\phi_\gamma$ is even, reversing an edge does not change its length, and $\boldsymbol{W}_e$ is symmetric. Moreover, $w_{ii}=0$ because the coordinate displacement of a pin from itself is zero. The matrix collects the geometric information needed by the subsequent topological computation, while Equation~\eqref{eq:edge_weight} retains the explicit coordinate dependence required for the backward pass.

\subsubsection{Differentiable Surrogate via LSE}
Each spanning tree $T\in\mathcal{T}_e$ contains exactly $n-1$ edges, and its smoothed length is the sum of the corresponding entries of $\boldsymbol{W}_e$. Minimizing this sum over $\mathcal{T}_e$ would still select a discrete tree and could introduce a change in gradient when the minimizing topology switches. Applying the smooth minimum in Equation~\eqref{eq:lse_def} gives the relaxed wirelength objective
\begin{equation}\label{eq:lse_obj}
    \hat{\mathcal{W}}_{\text{RMST}}^\beta(e;x,y)
       =-\beta\log\sum_{T\in\mathcal{T}_e}
          \exp\!\left(-\frac{1}{\beta}
              \sum_{\{u,v\}\in T}w_{uv}\right).
\end{equation}
Here, the inner sum evaluates one candidate tree, and the outer sum aggregates all candidate topologies. The parameter $\beta$ controls the concentration of this weighting. A smaller $\beta$ concentrates the contribution on shorter trees, whereas a larger $\beta$ allows a wider range of topologies to influence the objective and its gradient.

For fixed $\gamma$, the LSE convergence property discussed in Section~\ref{sec:preliminaries} implies convergence to $\min_{T\in\mathcal{T}_e}\sum_{\{u,v\}\in T}w_{uv}$ as $\beta\to0^+$. The remaining difference from the exact RMST then comes from the spatial approximation in Equation~\eqref{eq:edge_weight}. Although Equation~\eqref{eq:lse_obj} provides the desired differentiability, it is not yet a practical evaluation rule. A complete graph with $n\geq2$ vertices has $n^{n-2}$ spanning trees, so computing its terms individually would introduce an exponential bottleneck. We next evaluate this sum algebraically using the weighted matrix-tree theorem.

\subsection{Algebraic Evaluation via Matrix-Tree Theorem}
\label{sec:mtt_eval}
The summation inside the logarithm in Equation~\eqref{eq:lse_obj} is a positive scalar partition function:
\begin{equation}
    Z_e=\sum_{T\in\mathcal{T}_e}
         \exp\!\left(-\frac{1}{\beta}
             \sum_{\{u,v\}\in T}w_{uv}\right).
    \label{eq:partition_function}
\end{equation}
The weighted matrix-tree theorem evaluates this sum using matrices whose dimensions depend only on the number of pins. To apply it, the geometric edge lengths must first be converted into affinities whose products reproduce the exponential tree weights.

Because the edge weight matrix has a zero diagonal, this conversion has the compact form
\begin{equation}
    \boldsymbol{A}_e
       =\hat{\exp}\!\left(-\frac{\boldsymbol{W}_e}{\beta}\right)
        -\boldsymbol{I},
    \label{eq:affinity_matrix}
\end{equation}
where $\hat{\exp}(\cdot)$ denotes the element-wise exponential operator and $\boldsymbol{I}$ is the $n\times n$ identity matrix. Each off-diagonal affinity is $a_{ij}=\exp(-w_{ij}/\beta)>0$, so smaller geometric lengths correspond to stronger affinities. The degree matrix and Laplacian then follow from Equation~\eqref{eq:laplacian_matrices}. Since all candidate connections have positive affinity, the weighted graph is connected in exact arithmetic.

\begin{theorem}[Weighted Matrix-Tree Theorem]\label{thm:mtt}
Let $\boldsymbol{L}_e^k$ denote the $(n-1)\times(n-1)$ principal submatrix obtained by deleting the $k$-th row and column of $\boldsymbol{L}_e$. For the connected weighted graph $G_e$, the determinant of this minor equals the sum, over spanning trees, of the products of their edge affinities~\cite{chaiken1982allminors}:
\begin{equation}
    \det\!\left(\boldsymbol{L}_e^k\right)
       =\sum_{T\in\mathcal{T}_e} \Big( \prod_{\{u,v\}\in T}a_{uv} \Big).
    \label{eq:mtt}
\end{equation}
\end{theorem}
For each tree, substituting the affinities from Equation~\eqref{eq:affinity_matrix} gives
\begin{equation}
    \begin{aligned}
    \prod_{\{u,v\}\in T}a_{uv}
       &=\prod_{\{u,v\}\in T}\exp\!\left(-\frac{w_{uv}}{\beta}\right)\\
       &=\exp\!\left(-\frac{1}{\beta}
            \sum_{\{u,v\}\in T}w_{uv}\right).
    \end{aligned}
    \label{eq:tree_affinity_product}
\end{equation}
The final expression is precisely the contribution of that tree to $Z_e$. The theorem therefore computes the original partition function, rather than a separate approximation to it. This representation replaces tree enumeration with a factorization of the $(n-1)\times(n-1)$ minor~\cite{koo2007structured}. In particular, $Z_e=\det(\boldsymbol{L}_e^k)$, and the relaxed wirelength becomes
\begin{equation}
    \hat{\mathcal{W}}_{\text{RMST}}^\beta(e;x,y)
        =-\beta\log\det\!\left(\boldsymbol{L}_e^k\right).
    \label{eq:logdet_obj}
\end{equation}

\subsection{Analytical Gradient and Marginal Probability}
\label{sec:analytical_gradient}
The wirelength objective in Equation~\eqref{eq:logdet_obj} depends on pin coordinates through the smoothed edge lengths. Moving pin $i$ changes only its incident edges. Therefore, by the chain rule, the coordinate gradient can be expressed as the sum of the sensitivity of the objective to each incident edge weight multiplied by the corresponding spatial derivative:
\begin{equation}
    \frac{\partial\hat{\mathcal{W}}_{\text{RMST}}^\beta}{\partial x_i}
       =\sum_{\substack{j\in V_e\\j\neq i}}
          \frac{\partial\hat{\mathcal{W}}_{\text{RMST}}^\beta}{\partial w_{ij}}
          \frac{\partial w_{ij}}{\partial x_i}.
    \label{eq:chain_rule}
\end{equation}

For the first factor, we treat $w_{ij}=w_{ji}$ as one undirected edge length. A tree's length has derivative $1$ with respect to $w_{ij}$ if it contains $\{i,j\}$, and $0$ otherwise. Differentiating the LSE objective therefore gives
\begin{equation}
    \begin{aligned}
    \frac{\partial\hat{\mathcal{W}}_{\text{RMST}}^\beta}{\partial w_{ij}}
       &=-\beta\cdot\frac{1}{Z_e}\cdot\frac{\partial Z_e}{\partial w_{ij}}\\
       &=\frac{\sum_{\substack{T\in\mathcal{T}_e\\\{i,j\}\in T}}
          \exp\!\left(-\frac{1}{\beta}\sum_{\{u,v\}\in T}w_{uv}\right)}
          {\sum_{T\in\mathcal{T}_e}
          \exp\!\left(-\frac{1}{\beta}\sum_{\{u,v\}\in T}w_{uv}\right)}.
    \end{aligned}
    \label{eq:prob_derivation}
\end{equation}
The numerator includes only trees containing $\{i,j\}$, so this ratio is the edge's marginal probability under the normalized exponential weights. Denoting this probability by $p_{ij}$ gives
\begin{equation}
    p_{ij}\triangleq
       \Pr\!\left(\{i,j\}\in T\right)
       =\frac{\partial\hat{\mathcal{W}}_{\text{RMST}}^\beta}{\partial w_{ij}},
       \quad i\neq j.
    \label{eq:marginal_prob_def}
\end{equation}
Edges occurring frequently in shorter trees consequently have a greater influence on the gradient. Since $\{i,j\}$ and $\{j,i\}$ are the same edge, the marginal probability matrix is symmetric:
\begin{equation}
    \boldsymbol{P}_e=[p_{ij}]=\boldsymbol{P}_e^{\mathsf T}
       \in\mathbb{R}^{n\times n},
    \qquad p_{ii}=0.
    \label{eq:marginal_probability_matrix}
\end{equation}
As with $w_{ij}$, the net index is implicit in the scalar entries. Each $p_{ij}\in[0,1]$, and $\sum_{i<j}p_{ij}=n-1$ because every spanning tree contains exactly $n-1$ edges.

To evaluate these probabilities, we differentiate the weighted matrix-tree identity in Equation~\eqref{eq:mtt}. The symmetric entries $a_{ij}=a_{ji}$ represent one undirected affinity and vary together. Differentiation with respect to this affinity selects trees containing $\{i,j\}$, and multiplication by $a_{ij}$ restores their weights. Thus,
\begin{equation}
    \begin{aligned}
    p_{ij}
       &=\frac{a_{ij}}{Z_e}\frac{\partial Z_e}{\partial a_{ij}}\\
       &=a_{ij}\frac{\partial\log\det(\boldsymbol{L}_e^k)}{\partial a_{ij}}\\
       &=a_{ij}\operatorname{tr}\!\left[
          (\boldsymbol{L}_e^k)^{-1}
          \frac{\partial\boldsymbol{L}_e^k}{\partial a_{ij}}
         \right].
    \end{aligned}
    \label{eq:edge_affinity_derivative}
\end{equation}
The last equality follows from Jacobi's formula, where $\operatorname{tr}(\cdot)$ denotes the matrix trace. To expand this trace uniformly for all pin pairs, we choose $k=n$ and pad the inverse minor with a zero row and column corresponding to the removed pin:
\begin{equation}
    \boldsymbol{M}_e=
    \begin{bmatrix}
        (\boldsymbol{L}_e^k)^{-1}&\boldsymbol{0}\\
        \boldsymbol{0}^{\mathsf T}&0
    \end{bmatrix}\in\mathbb{R}^{n\times n}.
    \label{eq:expanded_inverse}
\end{equation}
Let $L_{e,uv}$ and $M_{e,uv}$ denote the corresponding matrix entries. The undirected affinity contributes to both endpoint degrees and both off-diagonal entries, giving $\partial L_{e,ii}/\partial a_{ij}=\partial L_{e,jj}/\partial a_{ij}=1$ and $\partial L_{e,ij}/\partial a_{ij}=\partial L_{e,ji}/\partial a_{ij}=-1$, with all other derivatives zero. Expanding the trace in Equation~\eqref{eq:edge_affinity_derivative} therefore gives
\begin{equation}
    \begin{aligned}
    p_{ij}
       &=a_{ij}\sum_{u,v\in V_e}
          M_{e,vu}\frac{\partial L_{e,uv}}{\partial a_{ij}}\\
       &=a_{ij}(M_{e,ii}+M_{e,jj}-M_{e,ij}-M_{e,ji})\\
       &=a_{ij}(M_{e,ii}+M_{e,jj}-2M_{e,ij}).
    \end{aligned}
    \label{eq:eff_resistance}
\end{equation}
The final equality uses the symmetry of $\boldsymbol{M}_e$. Zero padding makes the expression valid also for edges incident to the removed pin $k$, for which $p_{ik}=a_{ik}M_{e,ii}$. The resulting edge probabilities are independent of the choice of $k$.

Collecting these coefficients gives a matrix evaluation of Equation~\eqref{eq:marginal_probability_matrix}:
\begin{equation}
    \begin{aligned}
    \boldsymbol{P}_e
       &=\boldsymbol{A}_e\odot
         \bigl[\operatorname{Diag}(\boldsymbol{M}_e)\boldsymbol{1}^{\mathsf T}\\
       &\qquad+\boldsymbol{1}\operatorname{Diag}(\boldsymbol{M}_e)^{\mathsf T}
          -2\boldsymbol{M}_e\bigr],
    \end{aligned}
    \label{eq:edge_coefficients}
\end{equation}
where $\operatorname{Diag}(\boldsymbol{M}_e)$ is the column vector of diagonal entries and $\odot$ denotes the Hadamard product. The two outer products repeat the diagonal entries across rows and columns, respectively, producing a symmetric matrix with a zero diagonal.

Each entry $p_{ij}$ supplies the edge sensitivity in Equation~\eqref{eq:chain_rule}. The remaining factor is the spatial derivative, obtained from the smoothed edge length in Equation~\eqref{eq:edge_weight}:
\begin{equation}\label{eq:spatial_deriv}
    \begin{aligned}
    \frac{\partial w_{ij}}{\partial x_i}
        &=\tanh\!\left(\frac{x_i-x_j}{\gamma}\right)\\
        &\quad+\frac{x_i-x_j}{\gamma}\operatorname{sech}^2\!\left(\frac{x_i-x_j}{\gamma}\right),\\
    \frac{\partial w_{ij}}{\partial y_i}
        &=\tanh\!\left(\frac{y_i-y_j}{\gamma}\right)\\
        &\quad+\frac{y_i-y_j}{\gamma}\operatorname{sech}^2\!\left(\frac{y_i-y_j}{\gamma}\right).
    \end{aligned}
\end{equation}
Substituting these spatial derivatives and the edge marginal $p_{ij}$ into Equation~\eqref{eq:chain_rule} yields
\begin{equation}
    \begin{aligned}
    \frac{\partial\hat{\mathcal{W}}_{\text{RMST}}^\beta}{\partial x_i}
       &=\sum_{\substack{j\in V_e\\j\neq i}}
          p_{ij}\frac{\partial w_{ij}}{\partial x_i},\\
    \frac{\partial\hat{\mathcal{W}}_{\text{RMST}}^\beta}{\partial y_i}
       &=\sum_{\substack{j\in V_e\\j\neq i}}
          p_{ij}\frac{\partial w_{ij}}{\partial y_i},
    \end{aligned}
    \label{eq:coordinate_gradients}
\end{equation}
for every $i\in V_e$. Because each edge length depends on coordinate differences, it contributes equal and opposite gradients at its endpoints. A single inverse minor supplies all coefficients in $\boldsymbol{P}_e$, after which these gradients require only an accumulation over incident edges. The next subsection reuses $\boldsymbol{P}_e$ to obtain the forward output within the same computation.

\subsection{Fused Forward and Backward Operator Implementation}
\label{sec:parallel_merged}
Although the differentiable formulation enables gradient-based optimization, its repeated evaluation over a large number of nets can introduce substantial computational overhead in the global placement iteration. Thus, we fuse the wirelength forward value and the required spatial derivatives within each net's RMST evaluation. The CPU uses dynamic scheduling across nets. CUDA selects a single thread, a subgroup within a warp, or a thread block according to net degree. Small multipin nets use register-resident matrices and shuffles, while larger nets use shared memory, warp-owned rows, and consecutive-column accesses, with large workspaces reused across batches. Finally, the existing density overflow controls $\gamma$ and the degree-dependent $\beta$, maintaining smoothness early in placement and sharpening the edge distribution as density feasibility improves.

\subsubsection{Forward Cost Proxy}
Linearity of expectation gives the forward cost proxy from the edge marginals:
\begin{equation}
\mathbb{E}_e[\mathcal{W}_{\text{RMST}}]=\frac12\langle\boldsymbol{W}_e,\boldsymbol{P}_e\rangle=\sum_{\substack{i,j\in V_e\\i<j}}w_{ij}p_{ij}.
\label{eq:expected_forward}
\end{equation}
Here, $\langle\cdot,\cdot\rangle$ is the Frobenius inner product. The factor $1/2$ accounts for the two symmetric entries per undirected edge, so each unordered pair contributes once. The proxy is related to the LSE objective in Equation~\eqref{eq:logdet_obj} by
\begin{equation}
    \begin{aligned}
    \mathbb{E}_e[\mathcal{W}_{\text{RMST}}]
       -\hat{\mathcal{W}}_{\text{RMST}}^\beta
       &=\beta H_e,\\
    0\leq\beta H_e&\leq\beta\log|\mathcal{T}_e|,
    \end{aligned}
    \label{eq:forward_proxy_gap}
\end{equation}
where $H_e$ is the entropy of the normalized tree distribution. Corollary~\ref{cor:forward_proxy_error} in Section~\ref{sec:theory} proves this relation and bounds the proxy's error relative to exact RMST. Optimization uses the LSE derivatives in Equation~\eqref{eq:coordinate_gradients}, so replacing the reported forward value introduces no gradient approximation.

Equation~\eqref{eq:expected_forward} requires only a multiplication and an accumulation per edge during the existing gradient traversal. It avoids the separate determinant and logarithm evaluation for the scalar cost. The marginal solve is shared by both outputs and is not repeated in backward.

\begin{algorithm}[t]
    \caption{CPU Scheduling and Per-Net Computation}
    \label{alg:cpu_merged}
    \small
    \algrenewcommand\algorithmicindent{1em}
    \begin{algorithmic}[1]
        \Require Nets $E$, coordinates $x,y$, parameters $\{\beta_e,\gamma_e\}$
        \Ensure $\sum_e c_e,\{\boldsymbol{g}_i\}$
        \Function{TwoPin}{$e,x,y,\gamma$}
            \State $c_e\leftarrow w_{12}$
            \State $\boldsymbol{g}_1\leftarrow\nabla_{\mathbf{p}_1}w_{12}$, $\boldsymbol{g}_2\leftarrow-\boldsymbol{g}_1$
            \State \Return $c_e,\{\boldsymbol{g}_i\}_{i\in V_e}$
        \EndFunction
        \Function{MultiPin}{$e,x,y,\beta,\gamma$}
            \State $c_e\leftarrow0$, $\boldsymbol{g}_i\leftarrow\boldsymbol{0}\quad(i\in V_e)$
            \State $\boldsymbol{W}_e\leftarrow[w_{ij}]$ \Comment{joint edge derivatives}
            \State $\boldsymbol{A}_e\leftarrow\exp(-\boldsymbol{W}_e/\beta)-\boldsymbol{I}$
            \State $\boldsymbol{P}_e\leftarrow\Call{KronMarginals}{\boldsymbol{A}_e}$
            \ForAll{$i<j$}
                \State $c_e\leftarrow c_e+w_{ij}p_{ij}$
                \State $\boldsymbol{g}_i\leftarrow\boldsymbol{g}_i+p_{ij}\nabla_{\mathbf{p}_i}w_{ij}$
                \State $\boldsymbol{g}_j\leftarrow\boldsymbol{g}_j-p_{ij}\nabla_{\mathbf{p}_i}w_{ij}$
            \EndFor
            \State \Return $c_e,\{\boldsymbol{g}_i\}_{i\in V_e}$
        \EndFunction
        \Function{CPUForwardBackward}{$E,x,y,\{\beta_e,\gamma_e\}$}
            \State $c_e\leftarrow0$, $\boldsymbol{g}_i\leftarrow\boldsymbol{0}$ for all nets and pins
            \ForAll{nets $e$ in parallel} \Comment{dynamic scheduling}
                \If{$|V_e|=2$}
                    \State $c_e,\{\boldsymbol{g}_i\}_{i\in V_e}\leftarrow\Call{TwoPin}{e,x,y,\gamma_e}$
                \ElsIf{$|V_e|>2$}
                    \State $c_e,\{\boldsymbol{g}_i\}_{i\in V_e}\leftarrow\Call{MultiPin}{e,x,y,\beta_e,\gamma_e}$
                \EndIf
            \EndFor
            \State \Return $\sum_e c_e,\{\boldsymbol{g}_i\}$
        \EndFunction
    \end{algorithmic}
\end{algorithm}

\subsubsection{CPU Implementation}
Algorithm~\ref{alg:cpu_merged} gives the \textsc{TwoPin} and \textsc{MultiPin} routines followed by their parallel CPU driver. Here, $c_e$ denotes the forward cost proxy, and $\boldsymbol{g}_i=\bigl(\frac{\partial\hat{\mathcal{W}}_{\text{RMST}}^\beta}{\partial x_i},\frac{\partial\hat{\mathcal{W}}_{\text{RMST}}^\beta}{\partial y_i}\bigr)^{\mathsf T}$ is the LSE spatial gradient of pin $i$ in Equation~\eqref{eq:coordinate_gradients}. The per-net parameters $\beta_e,\gamma_e$ are specified after the complexity analysis, with their net indices omitted inside each per-net routine.

\textsc{TwoPin} evaluates one edge and its opposite endpoint derivatives without matrix operations. \textsc{MultiPin} evaluates the edge lengths and their spatial derivatives together, obtains $\boldsymbol{P}_e$ once, and accumulates both outputs over $i<j$. The common \textsc{KronMarginals} routine takes $\boldsymbol{A}_e$ as input and returns the symmetric $\boldsymbol{P}_e$ according to Equation~\eqref{eq:edge_coefficients}. Successive Kron elimination applies Schur complements to the Laplacian~\cite{dorfler2013kron}, and reverse reconstruction supplies the inverse information required by the edge marginals.

Dynamic scheduling assigns the next available net to each CPU worker, balancing the matrix workload across net degrees. Each worker reuses its private matrix storage. Distinct nets have disjoint pin sets, so their coordinate derivatives can be written independently.

\begin{algorithm}[t]
    \caption{GPU Computation for Small and Big Nets}
    \label{alg:gpu_kernels}
    \small
    \algrenewcommand\algorithmicindent{1em}
    \begin{algorithmic}[1]
        \Require Bucket $B$, coordinates $x,y$, parameters $\{\beta_e,\gamma_e\}$
        \Ensure $\{c_e\},\{\boldsymbol{g}_i\}$
        \Function{SmallNet}{$B,x,y,\{\beta_e,\gamma_e\},s$}
            \ForAll{$e\in B$ in parallel} \Comment{one $s$-lane subgroup}
                \State $\beta\leftarrow\beta_e$, $\gamma\leftarrow\gamma_e$
                \State $i\leftarrow1+\text{lane index within subgroup}$
                \State $\boldsymbol{W}_e[i,:]\leftarrow[w_{ij}]_j$ \Comment{register row}
                \State $\boldsymbol{A}_e[i,:]\leftarrow[a_{ij}]_j$ \Comment{shuffles}
                \State $\boldsymbol{P}_e\leftarrow\Call{KronMarginals}{\boldsymbol{A}_e}$
                \State $\boldsymbol{g}_i,c_{e,i}\leftarrow\text{Eqs.~\eqref{eq:coordinate_gradients},~\eqref{eq:expected_forward}}$
                \State $c_e\leftarrow\operatorname{Reduce}_s(c_{e,i})$ \Comment{lane $0$ writes}
            \EndFor
        \EndFunction
        \Function{BigNet}{$B,x,y,\{\beta_e,\gamma_e\},\text{storage}$}
            \ForAll{$e\in B$ in parallel} \Comment{one $128$-thread block}
                \State $\beta\leftarrow\beta_e$, $\gamma\leftarrow\gamma_e$
                \ForAll{owned pairs $i<j$ in parallel}
                    \State $\boldsymbol{W}_e[i,j]\leftarrow w_{ij}$
                    \State $\boldsymbol{A}_e[i,j]\leftarrow\exp(-w_{ij}/\beta)$
                \EndFor
                \State \Call{SyncBlock}{}
                \State $\boldsymbol{P}_e\leftarrow\Call{KronMarginals}{\boldsymbol{A}_e}$ \Comment{block sync}
                \ForAll{warp-owned rows $i$ in parallel}
                    \State $\boldsymbol{g}_i,c_{e,i}\leftarrow\text{Eqs.~\eqref{eq:coordinate_gradients},~\eqref{eq:expected_forward}}$
                \EndFor
                \State $c_e\leftarrow\operatorname{Reduce}_{\rm block}(c_{e,i})$
            \EndFor
        \EndFunction
    \end{algorithmic}
\end{algorithm}

\subsubsection{GPU Implementation}
The GPU reuses \textsc{TwoPin} with one thread per net. For $n>2$, \textsc{SmallNet} and \textsc{BigNet} in Algorithm~\ref{alg:gpu_kernels} map \textsc{MultiPin} to subgroups and thread blocks, respectively. Both evaluate Equations~\eqref{eq:coordinate_gradients} and~\eqref{eq:expected_forward} by rows, with $c_{e,i}$ denoting the cost contribution from pairs $(i,j)$ with $j>i$. Each pin's derivatives are accumulated by one lane or warp and written once, avoiding concurrent endpoint updates.

The GPU recomputes $w_{ij}$, $a_{ij}$, and the spatial derivatives during the fused output traversal. Thus, one matrix workspace successively holds affinities, elimination coefficients, and reconstructed inverse information, without separate dense arrays for edge derivatives and probabilities.

For a degree bucket $B$ containing nets with $3\leq n\leq32$, \textsc{SmallNet} assigns one net to a subgroup of width
\begin{equation}
    s=2^{\left\lceil\log_2\!\left(\max_{e\in B}|V_e|\right)\right\rceil}
       \in\{4,8,16,32\}.
    \label{eq:gpu_subgroup_width}
\end{equation}
Each valid lane owns one pin and one matrix row in registers. Unrolled row operations and subgroup shuffles provide the coordinate and elimination-data exchanges. A subgroup reduction combines the partial costs into $c_e$. Multiple nets share a warp when $s<32$. Padding lanes participate in collectives with neutral contributions, while only valid lanes evaluate pin rows and write outputs.

For $n>32$, \textsc{BigNet} assigns one $128$-thread block per net. Warps own rows, and consecutive lanes access consecutive columns. The upper triangle represents the symmetric $\boldsymbol{A}_e$ with a zero diagonal. The cubic elimination sweep updates only $i<j$, avoiding transposed shared-memory stores and index division during elimination. Reverse reconstruction restores symmetric entries before use. Warp reductions produce the coordinate derivatives, followed by a block reduction for $c_e$. Matrix storage and the $\mathcal{O}(n)$ auxiliary arrays reside in shared memory when they fit and otherwise use reusable global-memory batches.

Nets with $n\geq2$ are partitioned into degree buckets, each confined to one kernel range. Algorithm~\ref{alg:gpu_merged} dispatches these buckets sequentially on the current CUDA stream. Each iteration first selects the kernel and its execution configuration, then launches that kernel once with the bucket's coordinates and parameters. The \textsc{TwoPin} launch applies the per-net routine in Algorithm~\ref{alg:cpu_merged} to each net using one thread. The final scalar reduction and the subsequent backward operations execute on the same stream.

\begin{algorithm}[t]
    \caption{GPU Scheduling over Degree Buckets}
    \label{alg:gpu_merged}
    \small
    \algrenewcommand\algorithmicindent{1em}
    \begin{algorithmic}[1]
        \Require Degree buckets of $E$, coordinates $x,y$, $\{\beta_e,\gamma_e\}$
        \Ensure $\sum_e c_e,\{\boldsymbol{g}_i\}$
        \Function{GPUForwardBackward}{$E,x,y,\{\beta_e,\gamma_e\}$}
            \State $c_e\leftarrow0$, $\boldsymbol{g}_i\leftarrow\boldsymbol{0}$ for all nets and pins
            \For{each degree bucket $B$ sequentially}
                \If{$\max_{e\in B}|V_e|=2$}
                    \State \Call{Launch}{$\textsc{TwoPin},B,x,y,\{\gamma_e\}$}
                \ElsIf{$\max_{e\in B}|V_e|\leq32$}
                    \State $s\leftarrow\text{Equation~\eqref{eq:gpu_subgroup_width}}$
                    \State \Call{Launch}{$\textsc{SmallNet},B,x,y,\{\beta_e,\gamma_e\},s$}
                \Else
                    \State $\text{storage}\leftarrow\text{shared if it fits, else global}$
                    \State \Call{Launch}{$\textsc{BigNet},B,x,y,\{\beta_e,\gamma_e\},\text{storage}$}
                \EndIf
            \EndFor
            \State \Return $\sum_e c_e,\{\boldsymbol{g}_i\}$
        \EndFunction
    \end{algorithmic}
\end{algorithm}

\subsubsection{Complexity Analysis}
For an $n$-pin net, Kron elimination and reverse reconstruction require $\mathcal{O}(\sum_{k=1}^{n-1}(n-k)^2)=\mathcal{O}(n^3)$ operations. Edge evaluation and fused accumulation need $\mathcal{O}(n^2)$ work and temporary storage. The two-pin path requires constant work and storage.

The total work of fused evaluation ($\mathcal{C}_{\rm f}$) and backward ($\mathcal{C}_{\rm b}$) is
\begin{equation}
    \mathcal{C}_{\rm f}=\mathcal{O}\!\left(\sum_{e\in E}|V_e|^3\right),
    \qquad
    \mathcal{C}_{\rm b}=\mathcal{O}\!\left(\sum_{e\in E}|V_e|\right).
    \label{eq:merged_complexity}
\end{equation}
With $q$ CPU workers, dynamic task assignment gives the ideal execution-time bound
\begin{equation}
    T_{\rm CPU}
       =\mathcal{O}\!\left(\frac{\mathcal{C}_{\rm f}}{q}
          +\max_{e\in E}|V_e|^3\right),
    \label{eq:cpu_parallel_bound}
\end{equation}
excluding scheduling overhead. The second term accounts for the longest indivisible net computation. On the GPU, subgroup and block cooperation parallelize the row and column operations within each stage, while successive elimination and reconstruction stages remain dependent. Bucket execution times add because dispatch is sequential, with one launch per bucket.

Forward retains $2\sum_e|V_e|$ coordinate derivatives, which backward scales by the upstream gradient and weights without matrix operations. CPU temporary storage is $\mathcal{O}(q\max_e|V_e|^2)$, while the GPU workspace is sized for the largest batch. Fusion avoids repeated marginal evaluation without changing the cubic arithmetic complexity.

\subsubsection{Adaptive Parameter Update}
We adapt the net-specific smoothing parameters to density overflow and net degree. Let $o_t$ denote the overflow after global placement iteration $t$ and $o_{\rm stop}=0.1$ the density stopping threshold. We set $\gamma_{\rm base}=4$ and $\beta_0=1/6$, and initialize $\gamma_0=10\gamma_{\rm base}$ and $r_0=1$. The corresponding initial parameters of an $n$-pin net are $\gamma_{e,0}=5n\gamma_{\rm base}$ and $\beta_{e,0}=2/(3\log n)$. Subsequent iterations use
\begin{equation}
    \begin{aligned}
    &r_t\leftarrow\operatorname{clip}
       \!\left(\frac{o_t-o_{\rm stop}}{1-o_{\rm stop}},0,1\right),\\
    &\gamma_t\leftarrow\gamma_{\rm base}
       10^{\frac{20}{9}(o_t-o_{\rm stop})-1},\\
    &\gamma_{e,t}\leftarrow\gamma_t
       \left[1+r_t\left(\frac{n}{2}-1\right)\right],\\
    &\beta_{e,t}\leftarrow\beta_0\frac{1+3r_t}{\log n},
    \end{aligned}
    \label{eq:adaptive_parameters}
\end{equation}
where $n=|V_e|\geq2$ and $\operatorname{clip}(\cdot,0,1)$ truncates its argument to $[0,1]$. The update of $\gamma_t$ follows the DREAMPlace implementation~\cite{lin2021dreamplace}. The parameter $\beta$ follows the temperature convention $a_{ij}=\exp(-w_{ij}/\beta)$. Large overflow increases both spatial and topological smoothing. As overflow approaches $o_{\rm stop}$, the temperatures decrease and the edge distribution becomes more concentrated. An increase in overflow can restore smoothing. The resulting $\beta_e=\beta_{e,t}$ and $\gamma_e=\gamma_{e,t}$ remain fixed throughout each fused evaluation and its backward operation.

\subsection{Theoretical Properties}
\label{sec:theory}
Following the formulation and implementation, we analyze the theoretical properties relevant to placement optimization.
We next analyze the approximation error, invariance properties, and gradient regularity of Equation~\eqref{eq:lse_obj}. Throughout this subsection, $n\geq2$ and the per-net parameters $\beta,\gamma>0$ are held fixed, as in each fused evaluation. The forward expectation in Equation~\eqref{eq:expected_forward} is analyzed separately from the LSE objective whose gradient drives optimization.

\begin{theorem}[Bounded Approximation Error]
\label{thm:bounded_error}
For every pin configuration, the relaxed objective satisfies
\begin{equation}
    \begin{aligned}
    0&\leq\mathcal{W}_{\text{RMST}}
       -\hat{\mathcal{W}}_{\text{RMST}}^\beta\\
     &\leq2C\gamma(n-1)+\beta(n-2)\log n,
    \end{aligned}
    \label{eq:rmst_approximation_bound}
\end{equation}
where $C=\max_{u\geq0}u(1-\tanh u)\approx0.278465$.
\end{theorem}
\begin{IEEEproof}
The evenness of $\phi_\gamma$ and the definition of $C$ give
\begin{equation}
    0\leq |d|-\phi_\gamma(d)
      =\gamma\frac{|d|}{\gamma}
        \left(1-\tanh\frac{|d|}{\gamma}\right)
      \leq C\gamma.
    \label{eq:scalar_smoothing_bound}
\end{equation}
Thus $0\leq d_{ij}-w_{ij}\leq2C\gamma$. Every spanning tree contains $n-1$ edges, so its smoothed length underestimates its Manhattan length by at most $2C\gamma(n-1)$. Taking the minimum over the same set of spanning trees preserves the additive error bound. The LSE relaxation contributes the additional error
\begin{equation}
    \begin{aligned}
    0&\leq\min_{T\in\mathcal{T}_e}
        \sum_{\{i,j\}\in T}w_{ij}
        -\hat{\mathcal{W}}_{\text{RMST}}^\beta\\
     &\leq\beta\log|\mathcal{T}_e|
      =\beta(n-2)\log n.
    \end{aligned}
    \label{eq:topological_relaxation_bound}
\end{equation}
Adding the spatial and topological bounds proves Equation~\eqref{eq:rmst_approximation_bound}.
\end{IEEEproof}

\begin{corollary}[Forward Proxy Error]
\label{cor:forward_proxy_error}
The expected wirelength returned by the fused operator satisfies the entropy relation in Equation~\eqref{eq:forward_proxy_gap} and the error bound
\begin{equation}
    \begin{aligned}
    -2C\gamma(n-1)
       &\leq\mathbb{E}_e[\mathcal{W}_{\text{RMST}}]
           -\mathcal{W}_{\text{RMST}}\\
       &\leq\beta(n-2)\log n.
    \end{aligned}
    \label{eq:forward_proxy_error_bound}
\end{equation}
\end{corollary}
\begin{IEEEproof}
Taking the expected negative logarithm of the normalized tree probability gives $H_e=\mathbb{E}_e[\mathcal{W}_{\text{RMST}}]/\beta+\log Z_e$. Substituting $\hat{\mathcal{W}}_{\text{RMST}}^\beta=-\beta\log Z_e$ and using the entropy bounds $0\leq H_e\leq\log|\mathcal{T}_e|$~\cite{boyd2004convex} proves Equation~\eqref{eq:forward_proxy_gap}.

The expectation is at least the minimum smoothed tree length, which is at least $\mathcal{W}_{\text{RMST}}-2C\gamma(n-1)$ by the preceding proof. For the upper bound, Equation~\eqref{eq:forward_proxy_gap} and Theorem~\ref{thm:bounded_error} give $\mathbb{E}_e[\mathcal{W}_{\text{RMST}}]\leq\hat{\mathcal{W}}_{\text{RMST}}^\beta+\beta\log|\mathcal{T}_e|\leq\mathcal{W}_{\text{RMST}}+\beta\log|\mathcal{T}_e|$.
\end{IEEEproof}
Both outputs therefore converge uniformly to the exact RMST wirelength as $\beta,\gamma\to0^+$ for fixed $n$. At finite temperatures, the LSE objective is a lower bound, whereas the forward proxy can lie on either side of the exact wirelength. For a two-pin net, the topological error vanishes because there is only one spanning tree.

\begin{theorem}[Spatial Translation Invariance]
\label{thm:translation_invariance}
For any translation $\mathbf{\Delta}=(\Delta_x,\Delta_y)$, set $x'=x+\Delta_x\boldsymbol{1}$ and $y'=y+\Delta_y\boldsymbol{1}$. Then
\begin{equation}
    \begin{aligned}
    \hat{\mathcal{W}}_{\text{RMST}}^\beta(e;x',y')
       &=\hat{\mathcal{W}}_{\text{RMST}}^\beta(e;x,y),\\
    \boldsymbol{P}'_e&=\boldsymbol{P}_e,
       \qquad\boldsymbol{g}'_i=\boldsymbol{g}_i.
    \end{aligned}
    \label{eq:translation_invariance}
\end{equation}
The forward proxy is also invariant, with $c'_e=c_e$.
\end{theorem}
\begin{IEEEproof}
Translation preserves all coordinate differences, hence $w'_{ij}=w_{ij}$ and the normalized weights of all trees are unchanged. This preserves the LSE objective and every edge marginal. Equation~\eqref{eq:spatial_deriv} also depends only on coordinate differences, so Equation~\eqref{eq:coordinate_gradients} gives $\boldsymbol{g}'_i=\boldsymbol{g}_i$. The proxy equality follows from Equation~\eqref{eq:expected_forward}.
\end{IEEEproof}

\begin{corollary}[Zero-Sum Gradient]
\label{cor:zero_sum}
The net's wirelength gradients satisfy
\begin{equation}
    \sum_{i\in V_e}\boldsymbol{g}_i=\boldsymbol{0}.
    \label{eq:zero_sum_gradient}
\end{equation}
\end{corollary}
\begin{IEEEproof}
Each unordered edge contributes $p_{ij}\nabla_{\mathbf{p}_i}w_{ij}$ at pin $i$ and its negative at pin $j$. Summing over all pins cancels these contributions pairwise.
\end{IEEEproof}

\begin{theorem}[Permutation Invariance and Gradient Equivariance]
\label{thm:permutation_invariance}
For any permutation $\pi:V_e\to V_e$, relabel the pins by $\mathbf{p}'_i=\mathbf{p}_{\pi(i)}$. Then
\begin{equation}
    \begin{aligned}
    \hat{\mathcal{W}}_{\text{RMST}}^\beta(e;x',y')
       &=\hat{\mathcal{W}}_{\text{RMST}}^\beta(e;x,y),\\
    p'_{ij}&=p_{\pi(i)\pi(j)},
       \qquad\boldsymbol{g}'_i=\boldsymbol{g}_{\pi(i)}.
    \end{aligned}
    \label{eq:permutation_equivariance}
\end{equation}
The forward proxy is invariant, with $c'_e=c_e$.
\end{theorem}
\begin{IEEEproof}
Relabeling preserves edge lengths through $w'_{ij}=w_{\pi(i)\pi(j)}$ and induces a bijection on $\mathcal{T}_e$. Corresponding trees have identical lengths and normalized weights, preserving the partition function and permuting the edge marginals. The spatial derivatives transform under the same relabeling, so Equation~\eqref{eq:coordinate_gradients} gives the stated gradient equivariance. Summing each unordered edge once in Equation~\eqref{eq:expected_forward} preserves the forward proxy.
\end{IEEEproof}

\begin{theorem}[Globally Lipschitz Continuous Gradient]
\label{thm:lipschitz}
Let $\mathbf{v}=(x_1,\ldots,x_n,y_1,\ldots,y_n)^{\mathsf T}\in\mathbb{R}^{2n}$. The LSE objective is infinitely differentiable, and its coordinate gradient is globally $L$-Lipschitz continuous with
\begin{equation}
    L=n\left(\frac{2}{\gamma}+\frac{36}{25\beta}\right).
    \label{eq:gradient_lipschitz_constant}
\end{equation}
In particular, for any $\mathbf{u},\mathbf{v}\in\mathbb{R}^{2n}$,
\begin{equation}
    \begin{aligned}
    &\left\|\nabla\hat{\mathcal{W}}_{\text{RMST}}^\beta(\mathbf{u})
       -\nabla\hat{\mathcal{W}}_{\text{RMST}}^\beta(\mathbf{v})\right\|_2\\
    &\hspace{2em}\leq L\|\mathbf{u}-\mathbf{v}\|_2.
    \end{aligned}
    \label{eq:gradient_lipschitz}
\end{equation}
\end{theorem}
\begin{IEEEproof}
Smoothness follows from Equation~\eqref{eq:lse_obj}, a finite sum of positive smooth terms inside the logarithm. Differentiating $\phi_\gamma$ gives the uniform bounds
\begin{equation}
    |\phi_\gamma'(d)|\leq\frac65,
    \qquad |\phi_\gamma''(d)|\leq\frac2\gamma.
    \label{eq:spatial_derivative_bounds}
\end{equation}
Indeed, $\phi_\gamma''(d)=\frac2\gamma\operatorname{sech}^2u(1-u\tanh u)$ with $u=d/\gamma$. The maximum of $|\phi_\gamma'|$ is the positive root of $u\tanh u=1$, which is less than $6/5$, and $|\operatorname{sech}^2u(1-u\tanh u)|\leq1$.

Index each unordered edge once, using $i<j$. Let $\boldsymbol{J}$ have rows $(\nabla_{\mathbf{v}}w_{ij})^{\mathsf T}$, and let $\boldsymbol{\Sigma}$ be the covariance matrix of the indicators of $\{i,j\}\in T$ under the tree distribution in Equation~\eqref{eq:prob_derivation}. Differentiating the normalized exponential weights yields
\begin{equation}
    \nabla_{\mathbf{v}}^2\hat{\mathcal{W}}_{\text{RMST}}^\beta
       =\sum_{i<j}p_{ij}\nabla_{\mathbf{v}}^2w_{ij}
        -\frac1\beta\boldsymbol{J}^{\mathsf T}
          \boldsymbol{\Sigma}\boldsymbol{J}.
    \label{eq:coordinate_hessian}
\end{equation}
The minus sign follows from the smooth minimum convention. To bound the covariance, let $\boldsymbol{b}_{ij}\in\mathbb{R}^n$ have entries $+1$ at $i$, $-1$ at $j$, and zero elsewhere. Since $\Sigma_{ij,uv}=a_{uv}\,\partial p_{ij}/\partial a_{uv}$, differentiating Equation~\eqref{eq:eff_resistance} gives, for distinct unordered edges $\{i,j\}$ and $\{u,v\}$,
\begin{equation}
    \begin{aligned}
    \Sigma_{ij,uv}
       &=-a_{ij}a_{uv}
          (\boldsymbol{b}_{ij}^{\mathsf T}
           \boldsymbol{M}_e\boldsymbol{b}_{uv})^2\leq0,\\
    \Sigma_{ij,ij}&=p_{ij}(1-p_{ij})\leq\frac14.
    \end{aligned}
    \label{eq:edge_indicator_covariance}
\end{equation}
These signs express the pairwise negative correlation of weighted spanning-tree edges~\cite{lyons2003determinantal}. Since every tree contains $n-1$ edges, $\boldsymbol{\Sigma}\boldsymbol{1}=\boldsymbol{0}$. Its absolute row sums therefore satisfy
\begin{equation}
    \|\boldsymbol{\Sigma}\|_2
       \leq\max_{i<j}\sum_{u<v}|\Sigma_{ij,uv}|
       =2\max_{i<j}p_{ij}(1-p_{ij})\leq\frac12.
    \label{eq:edge_covariance_bound}
\end{equation}

For any $\boldsymbol{z}\in\mathbb{R}^n$, the complete-graph identity $\sum_{i<j}(z_i-z_j)^2=n\|\boldsymbol{z}\|_2^2-(\boldsymbol{1}^{\mathsf T}\boldsymbol{z})^2\leq n\|\boldsymbol{z}\|_2^2$, together with Equation~\eqref{eq:spatial_derivative_bounds}, gives
\begin{equation}
    \|\boldsymbol{J}\|_2^2\leq\frac{72n}{25},
    \qquad
    \left\|\sum_{i<j}p_{ij}\nabla_{\mathbf{v}}^2w_{ij}\right\|_2
       \leq\frac{2n}{\gamma}.
    \label{eq:spatial_matrix_bounds}
\end{equation}
The second inequality also uses $0\leq p_{ij}\leq1$ and the absence of mixed derivatives between $x$ and $y$ in each edge length. Applying these bounds to Equation~\eqref{eq:coordinate_hessian} yields $\|\nabla_{\mathbf{v}}^2\hat{\mathcal{W}}_{\text{RMST}}^\beta\|_2\leq2n/\gamma+36n/(25\beta)=L$. Integrating the Hessian along the segment from $\mathbf{u}$ to $\mathbf{v}$ proves Equation~\eqref{eq:gradient_lipschitz}.
\end{IEEEproof}
The Lipschitz continuity of the gradient bounds its variation with respect to the pin coordinates, providing a regular and well-behaved optimization landscape for gradient-based placement. This property supports stable first-order optimization of the proposed differentiable RMST objective.

%% file: include/Figure/forward.tex
\begin{figure*}[t]
\centering

\tikzset{every picture/.style={line width=0.75pt}} %
\resizebox{\hsize}{!}{

\begin{tikzpicture}[
    x=0.75pt,y=0.75pt,yscale=-1,xscale=1,
    every node/.style={font=\fontsize{12}{14}\selectfont,inner sep=0.75pt},
    forward heading/.style={font=\fontsize{12}{14}\selectfont\bfseries,
        text=white,align=center},
    forward group/.style={font=\fontsize{13}{15}\selectfont\bfseries,
        fill={rgb,255:red,205;green,231;blue,251},rounded corners=7.5pt,
        minimum height=18.75pt,inner xsep=5pt,inner ysep=2pt},
    forward note/.style={font=\fontsize{11}{13}\selectfont,align=center}
]

\draw  [color={rgb, 255:red, 49; green, 116; blue, 178 }  ,draw opacity=1 ][fill={rgb, 255:red, 237; green, 248; blue, 253 }  ,fill opacity=1 ][dash pattern={on 4.5pt off 4.5pt}] (127.17,279.48) -- (371.67,279.48) .. controls (380.18,279.48) and (387.44,274.17) .. (390.33,266.68) -- (390.33,266.5) .. controls (390.33,255.45) and (399.29,246.5) .. (410.33,246.5) -- (716.33,246.5) .. controls (727.38,246.5) and (736.33,255.45) .. (736.33,266.5) .. controls (736.33,266.5) and (736.33,372.01) .. (736.33,372.01) -- (736.33,496) .. controls (736.33,507.05) and (727.38,516) .. (716.33,516) -- (20,516) .. controls (8.95,516) and (0,507.05) .. (0,496) -- (0,350) .. controls (0,338.95) and (8.95,330) .. (20,330) -- (91.71,330) .. controls (99.9,328.16) and (106.19,321.28) .. (107.17,312.79) -- (107.17,299.48) .. controls (107.17,288.44) and (116.12,279.48) .. (127.17,279.48) -- cycle ;
\draw  [color={rgb, 255:red, 49; green, 116; blue, 178 }  ,draw opacity=1 ][fill={rgb, 255:red, 230; green, 250; blue, 253 }  ,fill opacity=1 ][dash pattern={on 4.5pt off 4.5pt}] (1006,496.32) .. controls (1006,507.47) and (997.05,516.5) .. (986,516.5) -- (772,516.5) .. controls (760.95,516.5) and (752,507.47) .. (752,496.32) -- (752,251) .. controls (749.73,242.37) and (741.94,236.01) .. (732.67,236.01) -- (359,236.01) .. controls (347.95,236.01) and (339,226.97) .. (339,215.83) -- (339,30.18) .. controls (339,19.03) and (347.95,10) .. (359,10) -- (986,10) .. controls (997.05,10) and (1006,19.03) .. (1006,30.18) -- (1006,496.32) -- cycle ;
\draw  [color={rgb, 255:red, 49; green, 116; blue, 178 }  ,draw opacity=1 ][fill={rgb, 255:red, 235; green, 248; blue, 253 }  ,fill opacity=1 ][dash pattern={on 4.5pt off 4.5pt}] (0,30.5) .. controls (0,19.45) and (8.95,10.5) .. (20,10.5) -- (310,10.5) .. controls (321.05,10.5) and (330,19.45) .. (330,30.5) -- (330,230.33) .. controls (330,241.38) and (321.05,250.33) .. (310,250.33) -- (20,250.33) .. controls (8.95,250.33) and (0,241.38) .. (0,230.33) -- cycle ;
\draw  [fill={rgb, 255:red, 19; green, 97; blue, 157 }  ,fill opacity=1 ] (25,32) -- (135,32) .. controls (143.28,32) and (150,38.72) .. (150,47) -- (150,58) -- (10,58) -- (10,47) .. controls (10,38.72) and (16.72,32) .. (25,32) -- cycle ;
\draw  [fill={rgb, 255:red, 255; green, 255; blue, 255 }  ,fill opacity=1 ] (150,217) .. controls (150,225.28) and (143.28,232) .. (135,232) -- (25,232) .. controls (16.72,232) and (10,225.28) .. (10,217) -- (10,58) -- (150,58) -- (150,217) -- cycle ;
\draw  [color={rgb, 255:red, 0; green, 0; blue, 0 }  ,draw opacity=1 ][fill={rgb, 255:red, 183; green, 209; blue, 226 }  ,fill opacity=1 ] (27.67,139.4) -- (47.67,139.4) -- (47.67,170.73) -- (27.67,170.73) -- cycle ;
\draw  [color={rgb, 255:red, 0; green, 0; blue, 0 }  ,draw opacity=1 ][fill={rgb, 255:red, 183; green, 209; blue, 226 }  ,fill opacity=1 ] (70.17,109.4) -- (96.67,109.4) -- (96.67,139.9) -- (70.17,139.9) -- cycle ;
\draw  [color={rgb, 255:red, 0; green, 0; blue, 0 }  ,draw opacity=1 ][fill={rgb, 255:red, 183; green, 209; blue, 226 }  ,fill opacity=1 ] (112.67,85.4) -- (132.67,85.4) -- (132.67,115.4) -- (112.67,115.4) -- cycle ;
\draw  [color={rgb, 255:red, 0; green, 0; blue, 0 }  ,draw opacity=1 ][fill={rgb, 255:red, 183; green, 209; blue, 226 }  ,fill opacity=1 ] (79,170.73) -- (109,170.73) -- (109,190.73) -- (79,190.73) -- cycle ;
\draw  [color={rgb, 255:red, 49; green, 116; blue, 178 }  ,draw opacity=1 ][fill={rgb, 255:red, 49; green, 116; blue, 178 }  ,fill opacity=1 ] (107,96.23) .. controls (107,94.85) and (108.12,93.73) .. (109.5,93.73) .. controls (110.88,93.73) and (112,94.85) .. (112,96.23) .. controls (112,97.61) and (110.88,98.73) .. (109.5,98.73) .. controls (108.12,98.73) and (107,97.61) .. (107,96.23) -- cycle ;
\draw  [color={rgb, 255:red, 49; green, 116; blue, 178 }  ,draw opacity=1 ][fill={rgb, 255:red, 49; green, 116; blue, 178 }  ,fill opacity=1 ] (86,142.9) .. controls (86,141.52) and (87.12,140.4) .. (88.5,140.4) .. controls (89.88,140.4) and (91,141.52) .. (91,142.9) .. controls (91,144.28) and (89.88,145.4) .. (88.5,145.4) .. controls (87.12,145.4) and (86,144.28) .. (86,142.9) -- cycle ;
\draw  [color={rgb, 255:red, 49; green, 116; blue, 178 }  ,draw opacity=1 ][fill={rgb, 255:red, 49; green, 116; blue, 178 }  ,fill opacity=1 ] (48,152.9) .. controls (48,151.52) and (49.12,150.4) .. (50.5,150.4) .. controls (51.88,150.4) and (53,151.52) .. (53,152.9) .. controls (53,154.28) and (51.88,155.4) .. (50.5,155.4) .. controls (49.12,155.4) and (48,154.28) .. (48,152.9) -- cycle ;
\draw  [color={rgb, 255:red, 49; green, 116; blue, 178 }  ,draw opacity=1 ][fill={rgb, 255:red, 49; green, 116; blue, 178 }  ,fill opacity=1 ] (96,167.9) .. controls (96,166.52) and (97.12,165.4) .. (98.5,165.4) .. controls (99.88,165.4) and (101,166.52) .. (101,167.9) .. controls (101,169.28) and (99.88,170.4) .. (98.5,170.4) .. controls (97.12,170.4) and (96,169.28) .. (96,167.9) -- cycle ;

\draw [color={rgb, 255:red, 49; green, 116; blue, 178 }  ,draw opacity=1 ][line width=1.5]    (150,132) -- (176,132) ;
\draw [shift={(180,132)}, rotate = 180] [fill={rgb, 255:red, 49; green, 116; blue, 178 }  ,fill opacity=1 ][line width=0.08]  [draw opacity=0] (12.77,-6.13) -- (0,0) -- (12.77,6.13) -- cycle    ;
\draw [color={rgb, 255:red, 49; green, 116; blue, 178 }  ,draw opacity=1 ][line width=1.5]    (320,137) .. controls (364.11,208.97) and (292.72,238.7) .. (412.18,301.06) ;
\draw [shift={(414,302)}, rotate = 207.31] [fill={rgb, 255:red, 49; green, 116; blue, 178 }  ,fill opacity=1 ][line width=0.08]  [draw opacity=0] (11.61,-5.58) -- (0,0) -- (11.61,5.58) -- cycle    ;
\draw [color={rgb, 255:red, 49; green, 116; blue, 178 }  ,draw opacity=1 ][line width=1.5]    (907,257.75) .. controls (903.24,278.2) and (898.6,284.06) .. (885.04,298.46) ;
\draw [shift={(882.33,301.33)}, rotate = 313.37] [fill={rgb, 255:red, 49; green, 116; blue, 178 }  ,fill opacity=1 ][line width=0.08]  [draw opacity=0] (12.77,-6.13) -- (0,0) -- (12.77,6.13) -- cycle    ;
\draw [color={rgb, 255:red, 49; green, 116; blue, 178 }  ,draw opacity=1 ][line width=1.5]    (83,352.25) .. controls (67.1,341.5) and (63.01,332.93) .. (59.69,315.31) ;
\draw [shift={(59,311.5)}, rotate = 80.19] [fill={rgb, 255:red, 49; green, 116; blue, 178 }  ,fill opacity=1 ][line width=0.08]  [draw opacity=0] (11.61,-5.58) -- (0,0) -- (11.61,5.58) -- cycle    ;
\draw [color={rgb, 255:red, 49; green, 116; blue, 178 }  ,draw opacity=1 ][line width=1.5]    (493,132) -- (519,132) ;
\draw [shift={(523,132)}, rotate = 180] [fill={rgb, 255:red, 49; green, 116; blue, 178 }  ,fill opacity=1 ][line width=0.08]  [draw opacity=0] (11.61,-5.58) -- (0,0) -- (11.61,5.58) -- cycle    ;
\draw [color={rgb, 255:red, 49; green, 116; blue, 178 }  ,draw opacity=1 ][line width=1.5]    (320,137) -- (349.67,137.15) ;
\draw [shift={(353.67,137.17)}, rotate = 180.28] [fill={rgb, 255:red, 49; green, 116; blue, 178 }  ,fill opacity=1 ][line width=0.08]  [draw opacity=0] (11.61,-5.58) -- (0,0) -- (11.61,5.58) -- cycle    ;
\draw  [fill={rgb, 255:red, 255; green, 255; blue, 255 }  ,fill opacity=1 ] (988,487) .. controls (988,495.28) and (981.28,502) .. (973,502) -- (783,502) .. controls (774.72,502) and (768,495.28) .. (768,487) -- (768,328) -- (988,328) -- (988,487) -- cycle ;
\draw  [fill={rgb, 255:red, 19; green, 97; blue, 157 }  ,fill opacity=1 ] (783,302) -- (973,302) .. controls (981.28,302) and (988,308.72) .. (988,317) -- (988,328) -- (768,328) -- (768,317) .. controls (768,308.72) and (774.72,302) .. (783,302) -- cycle ;

\draw  [color={rgb, 255:red, 19; green, 97; blue, 157 }  ,draw opacity=1 ][fill={rgb, 255:red, 255; green, 255; blue, 255 }  ,fill opacity=1 ] (838,361.67) -- (858,361.67) -- (858,381.67) -- (838,381.67) -- cycle ;
\draw  [color={rgb, 255:red, 19; green, 97; blue, 157 }  ,draw opacity=1 ][fill={rgb, 255:red, 185; green, 220; blue, 250 }  ,fill opacity=1 ] (858,361.67) -- (878,361.67) -- (878,381.67) -- (858,381.67) -- cycle ;
\draw  [color={rgb, 255:red, 19; green, 97; blue, 157 }  ,draw opacity=1 ][fill={rgb, 255:red, 140; green, 197; blue, 242 }  ,fill opacity=1 ] (878,361.67) -- (898,361.67) -- (898,381.67) -- (878,381.67) -- cycle ;
\draw  [color={rgb, 255:red, 19; green, 97; blue, 157 }  ,draw opacity=1 ][fill={rgb, 255:red, 47; green, 134; blue, 209 }  ,fill opacity=1 ] (898,361.67) -- (918,361.67) -- (918,381.67) -- (898,381.67) -- cycle ;
\draw  [color={rgb, 255:red, 19; green, 97; blue, 157 }  ,draw opacity=1 ][fill={rgb, 255:red, 185; green, 220; blue, 250 }  ,fill opacity=1 ] (838,381.67) -- (858,381.67) -- (858,401.67) -- (838,401.67) -- cycle ;
\draw  [color={rgb, 255:red, 19; green, 97; blue, 157 }  ,draw opacity=1 ][fill={rgb, 255:red, 255; green, 255; blue, 255 }  ,fill opacity=1 ] (858,381.67) -- (878,381.67) -- (878,401.67) -- (858,401.67) -- cycle ;
\draw  [color={rgb, 255:red, 19; green, 97; blue, 157 }  ,draw opacity=1 ][fill={rgb, 255:red, 140; green, 197; blue, 242 }  ,fill opacity=1 ] (878,381.67) -- (898,381.67) -- (898,401.67) -- (878,401.67) -- cycle ;
\draw  [color={rgb, 255:red, 19; green, 97; blue, 157 }  ,draw opacity=1 ][fill={rgb, 255:red, 185; green, 220; blue, 250 }  ,fill opacity=1 ] (898,381.67) -- (918,381.67) -- (918,401.67) -- (898,401.67) -- cycle ;
\draw  [color={rgb, 255:red, 19; green, 97; blue, 157 }  ,draw opacity=1 ][fill={rgb, 255:red, 140; green, 197; blue, 242 }  ,fill opacity=1 ] (838,401.67) -- (858,401.67) -- (858,421.67) -- (838,421.67) -- cycle ;
\draw  [color={rgb, 255:red, 19; green, 97; blue, 157 }  ,draw opacity=1 ][fill={rgb, 255:red, 140; green, 197; blue, 242 }  ,fill opacity=1 ] (858,401.67) -- (878,401.67) -- (878,421.67) -- (858,421.67) -- cycle ;
\draw  [color={rgb, 255:red, 19; green, 97; blue, 157 }  ,draw opacity=1 ][fill={rgb, 255:red, 255; green, 255; blue, 255 }  ,fill opacity=1 ] (878,401.67) -- (898,401.67) -- (898,421.67) -- (878,421.67) -- cycle ;
\draw  [color={rgb, 255:red, 19; green, 97; blue, 157 }  ,draw opacity=1 ][fill={rgb, 255:red, 15; green, 95; blue, 174 }  ,fill opacity=1 ] (898,401.67) -- (918,401.67) -- (918,421.67) -- (898,421.67) -- cycle ;
\draw  [color={rgb, 255:red, 19; green, 97; blue, 157 }  ,draw opacity=1 ][fill={rgb, 255:red, 47; green, 134; blue, 209 }  ,fill opacity=1 ] (838,421.67) -- (858,421.67) -- (858,441.67) -- (838,441.67) -- cycle ;
\draw  [color={rgb, 255:red, 19; green, 97; blue, 157 }  ,draw opacity=1 ][fill={rgb, 255:red, 185; green, 220; blue, 250 }  ,fill opacity=1 ] (858,421.67) -- (878,421.67) -- (878,441.67) -- (858,441.67) -- cycle ;
\draw  [color={rgb, 255:red, 19; green, 97; blue, 157 }  ,draw opacity=1 ][fill={rgb, 255:red, 15; green, 95; blue, 174 }  ,fill opacity=1 ] (878,421.67) -- (898,421.67) -- (898,441.67) -- (878,441.67) -- cycle ;
\draw  [color={rgb, 255:red, 19; green, 97; blue, 157 }  ,draw opacity=1 ][fill={rgb, 255:red, 255; green, 255; blue, 255 }  ,fill opacity=1 ] (898,421.67) -- (918,421.67) -- (918,441.67) -- (898,441.67) -- cycle ;
\draw  [color={rgb, 255:red, 19; green, 97; blue, 157 }  ,draw opacity=1 ][line width=1.5]  (838,361.67) -- (918,361.67) -- (918,441.67) -- (838,441.67) -- cycle ;

\draw [color={rgb, 255:red, 49; green, 116; blue, 178 }  ,draw opacity=1 ][line width=1.5]    (225,462.5) .. controls (210.98,453.92) and (207.21,452.8) .. (193.16,437.26) ;
\draw [shift={(190.6,434.4)}, rotate = 48.48] [fill={rgb, 255:red, 49; green, 116; blue, 178 }  ,fill opacity=1 ][line width=0.08]  [draw opacity=0] (12.77,-6.13) -- (0,0) -- (12.77,6.13) -- cycle    ;
\draw [color={rgb, 255:red, 49; green, 116; blue, 178 }  ,draw opacity=1 ][line width=1.5]    (59,270.5) .. controls (62.27,257.18) and (66.42,250.19) .. (76.75,235.65) ;
\draw [shift={(79,232.5)}, rotate = 125.62] [fill={rgb, 255:red, 49; green, 116; blue, 178 }  ,fill opacity=1 ][line width=0.08]  [draw opacity=0] (12.77,-6.13) -- (0,0) -- (12.77,6.13) -- cycle    ;
\draw [color={rgb, 255:red, 49; green, 116; blue, 178 }  ,draw opacity=1 ][line width=1.5]    (493,132) .. controls (510.33,171.67) and (494,207.5) .. (536,221.5) .. controls (578,235.5) and (717,224.5) .. (756,236.5) .. controls (794.03,248.2) and (807.65,262.28) .. (834.88,297.73) ;
\draw [shift={(837,300.5)}, rotate = 232.6] [fill={rgb, 255:red, 49; green, 116; blue, 178 }  ,fill opacity=1 ][line width=0.08]  [draw opacity=0] (12.77,-6.13) -- (0,0) -- (12.77,6.13) -- cycle    ;
\draw [color={rgb, 255:red, 49; green, 116; blue, 178 }  ,draw opacity=1 ][line width=1.5]    (793,130) -- (819,130) ;
\draw [shift={(823,130)}, rotate = 180] [fill={rgb, 255:red, 49; green, 116; blue, 178 }  ,fill opacity=1 ][line width=0.08]  [draw opacity=0] (11.61,-5.58) -- (0,0) -- (11.61,5.58) -- cycle    ;
\draw  [fill={rgb, 255:red, 255; green, 255; blue, 255 }  ,fill opacity=1 ] (793,197) .. controls (793,205.28) and (786.28,212) .. (778,212) -- (538,212) .. controls (529.72,212) and (523,205.28) .. (523,197) -- (523,58) -- (793,58) -- (793,197) -- cycle ;
\draw  [color={rgb, 255:red, 19; green, 97; blue, 157 }  ,draw opacity=1 ][fill={rgb, 255:red, 15; green, 95; blue, 174 }  ,fill opacity=1 ] (552,89) -- (572,89) -- (572,109) -- (552,109) -- cycle ;
\draw  [color={rgb, 255:red, 19; green, 97; blue, 157 }  ,draw opacity=1 ][fill={rgb, 255:red, 90; green, 170; blue, 230 }  ,fill opacity=1 ] (572,89) -- (592,89) -- (592,109) -- (572,109) -- cycle ;
\draw  [color={rgb, 255:red, 19; green, 97; blue, 157 }  ,draw opacity=1 ][fill={rgb, 255:red, 140; green, 197; blue, 242 }  ,fill opacity=1 ] (592,89) -- (612,89) -- (612,109) -- (592,109) -- cycle ;
\draw  [color={rgb, 255:red, 19; green, 97; blue, 157 }  ,draw opacity=1 ][fill={rgb, 255:red, 140; green, 197; blue, 242 }  ,fill opacity=1 ] (612,89) -- (632,89) -- (632,109) -- (612,109) -- cycle ;
\draw  [color={rgb, 255:red, 19; green, 97; blue, 157 }  ,draw opacity=1 ][fill={rgb, 255:red, 47; green, 134; blue, 209 }  ,fill opacity=1 ] (572,109) -- (592,109) -- (592,129) -- (572,129) -- cycle ;
\draw  [color={rgb, 255:red, 19; green, 97; blue, 157 }  ,draw opacity=1 ][fill={rgb, 255:red, 185; green, 220; blue, 250 }  ,fill opacity=1 ] (592,109) -- (612,109) -- (612,129) -- (592,129) -- cycle ;
\draw  [color={rgb, 255:red, 19; green, 97; blue, 157 }  ,draw opacity=1 ][fill={rgb, 255:red, 220; green, 238; blue, 255 }  ,fill opacity=1 ] (612,109) -- (632,109) -- (632,129) -- (612,129) -- cycle ;
\draw  [color={rgb, 255:red, 19; green, 97; blue, 157 }  ,draw opacity=1 ][fill={rgb, 255:red, 47; green, 134; blue, 209 }  ,fill opacity=1 ] (592,129) -- (612,129) -- (612,149) -- (592,149) -- cycle ;
\draw  [color={rgb, 255:red, 19; green, 97; blue, 157 }  ,draw opacity=1 ][fill={rgb, 255:red, 140; green, 197; blue, 242 }  ,fill opacity=1 ] (612,129) -- (632,129) -- (632,149) -- (612,149) -- cycle ;
\draw  [color={rgb, 255:red, 19; green, 97; blue, 157 }  ,draw opacity=1 ][fill={rgb, 255:red, 47; green, 134; blue, 209 }  ,fill opacity=1 ] (612,149) -- (632,149) -- (632,169) -- (612,169) -- cycle ;
\draw  [color={rgb, 255:red, 19; green, 97; blue, 157 }  ,draw opacity=1 ][fill={rgb, 255:red, 90; green, 170; blue, 230 }  ,fill opacity=1 ] (552,109) -- (572,109) -- (572,129) -- (552,129) -- cycle ;
\draw  [color={rgb, 255:red, 19; green, 97; blue, 157 }  ,draw opacity=1 ][fill={rgb, 255:red, 140; green, 197; blue, 242 }  ,fill opacity=1 ] (552,129) -- (572,129) -- (572,149) -- (552,149) -- cycle ;
\draw  [color={rgb, 255:red, 19; green, 97; blue, 157 }  ,draw opacity=1 ][fill={rgb, 255:red, 140; green, 197; blue, 242 }  ,fill opacity=1 ] (552,149) -- (572,149) -- (572,169) -- (552,169) -- cycle ;
\draw  [color={rgb, 255:red, 19; green, 97; blue, 157 }  ,draw opacity=1 ][fill={rgb, 255:red, 185; green, 220; blue, 250 }  ,fill opacity=1 ] (572,129) -- (592,129) -- (592,149) -- (572,149) -- cycle ;
\draw  [color={rgb, 255:red, 19; green, 97; blue, 157 }  ,draw opacity=1 ][fill={rgb, 255:red, 220; green, 238; blue, 255 }  ,fill opacity=1 ] (572,149) -- (592,149) -- (592,169) -- (572,169) -- cycle ;
\draw  [color={rgb, 255:red, 19; green, 97; blue, 157 }  ,draw opacity=1 ][fill={rgb, 255:red, 140; green, 197; blue, 242 }  ,fill opacity=1 ] (592,149) -- (612,149) -- (612,169) -- (592,169) -- cycle ;

\draw  [color={rgb, 255:red, 19; green, 97; blue, 157 }  ,draw opacity=1 ][line width=1.5]  (552,89) -- (632,89) -- (632,169) -- (552,169) -- cycle ;

\path  [color={rgb, 255:red, 19; green, 97; blue, 157 }  ,draw opacity=1 ][fill={rgb, 255:red, 15; green, 95; blue, 174 }  ,fill opacity=1 ] (680,90.5) -- (700,90.5) -- (700,110.5) -- (680,110.5) -- cycle ;
\path  [color={rgb, 255:red, 19; green, 97; blue, 157 }  ,draw opacity=1 ][fill={rgb, 255:red, 90; green, 170; blue, 230 }  ,fill opacity=1 ] (700,90.5) -- (720,90.5) -- (720,110.5) -- (700,110.5) -- cycle ;
\path  [color={rgb, 255:red, 19; green, 97; blue, 157 }  ,draw opacity=1 ][fill={rgb, 255:red, 140; green, 197; blue, 242 }  ,fill opacity=1 ] (720,90.5) -- (740,90.5) -- (740,110.5) -- (720,110.5) -- cycle ;
\path  [color={rgb, 255:red, 19; green, 97; blue, 157 }  ,draw opacity=1 ][fill={rgb, 255:red, 47; green, 134; blue, 209 }  ,fill opacity=1 ] (700,110.5) -- (720,110.5) -- (720,130.5) -- (700,130.5) -- cycle ;
\path  [color={rgb, 255:red, 19; green, 97; blue, 157 }  ,draw opacity=1 ][fill={rgb, 255:red, 185; green, 220; blue, 250 }  ,fill opacity=1 ] (720,110.5) -- (740,110.5) -- (740,130.5) -- (720,130.5) -- cycle ;
\path  [color={rgb, 255:red, 19; green, 97; blue, 157 }  ,draw opacity=1 ][fill={rgb, 255:red, 47; green, 134; blue, 209 }  ,fill opacity=1 ] (720,130.5) -- (740,130.5) -- (740,150.5) -- (720,150.5) -- cycle ;
\path  [color={rgb, 255:red, 19; green, 97; blue, 157 }  ,draw opacity=1 ][fill={rgb, 255:red, 90; green, 170; blue, 230 }  ,fill opacity=1 ] (680,110.5) -- (700,110.5) -- (700,130.5) -- (680,130.5) -- cycle ;
\path  [color={rgb, 255:red, 19; green, 97; blue, 157 }  ,draw opacity=1 ][fill={rgb, 255:red, 140; green, 197; blue, 242 }  ,fill opacity=1 ] (680,130.5) -- (700,130.5) -- (700,150.5) -- (680,150.5) -- cycle ;
\path  [color={rgb, 255:red, 19; green, 97; blue, 157 }  ,draw opacity=1 ][fill={rgb, 255:red, 185; green, 220; blue, 250 }  ,fill opacity=1 ] (700,130.5) -- (720,130.5) -- (720,150.5) -- (700,150.5) -- cycle ;
\begin{scope}[color={rgb,255:red,19;green,97;blue,157},line width=1pt]
    \foreach \x in {680,700,720}
        \draw (\x,90.5) -- (\x,150.5);
    \foreach \y in {90.5,110.5,130.5}
        \draw (680,\y) -- (740,\y);
    \begin{scope}[dash pattern=on 2pt off 1.5pt]
        \foreach \x in {680,700,720}
            \draw (\x,150.5) -- (\x,170.5);
        \foreach \y in {90.5,110.5,130.5}
            \draw (740,\y) -- (760,\y);
        \draw (740,90.5) -- (740,170.5)
              (760,90.5) -- (760,170.5)
              (680,150.5) -- (760,150.5)
              (680,170.5) -- (760,170.5);
    \end{scope}
\end{scope}

\draw [color={rgb, 255:red, 49; green, 116; blue, 178 }  ,draw opacity=1 ][line width=1.5]    (641,129) -- (667,129) ;
\draw [shift={(671,129)}, rotate = 180] [fill={rgb, 255:red, 49; green, 116; blue, 178 }  ,fill opacity=1 ][line width=0.08]  [draw opacity=0] (11.61,-5.58) -- (0,0) -- (11.61,5.58) -- cycle    ;
\draw  [fill={rgb, 255:red, 19; green, 97; blue, 157 }  ,fill opacity=1 ] (538,32) -- (778,32) .. controls (786.28,32) and (793,38.72) .. (793,47) -- (793,58) -- (523,58) -- (523,47) .. controls (523,38.72) and (529.72,32) .. (538,32) -- cycle ;
\draw  [fill={rgb, 255:red, 19; green, 97; blue, 157 }  ,fill opacity=1 ] (429,254.5) -- (709,254.5) .. controls (717.28,254.5) and (724,261.22) .. (724,269.5) -- (724,280.5) -- (414,280.5) -- (414,269.5) .. controls (414,261.22) and (420.72,254.5) .. (429,254.5) -- cycle ;
\draw  [fill={rgb, 255:red, 255; green, 255; blue, 255 }  ,fill opacity=1 ] (724,389.5) .. controls (724,397.78) and (717.28,404.5) .. (709,404.5) -- (429,404.5) .. controls (420.72,404.5) and (414,397.78) .. (414,389.5) -- (414,280.5) -- (724,280.5) -- (724,389.5) -- cycle ;
\draw  [fill={rgb, 255:red, 19; green, 97; blue, 157 }  ,fill opacity=1 ] (240,415) -- (710,415) .. controls (718.28,415) and (725,421.72) .. (725,430) -- (725,441) -- (225,441) -- (225,430) .. controls (225,421.72) and (231.72,415) .. (240,415) -- cycle ;
\draw  [fill={rgb, 255:red, 255; green, 255; blue, 255 }  ,fill opacity=1 ] (725,492) .. controls (725,500.28) and (718.28,507) .. (710,507) -- (240,507) .. controls (231.72,507) and (225,500.28) .. (225,492) -- (225,441) -- (725,441) -- (725,492) -- cycle ;
\draw  [color={rgb, 255:red, 0; green, 0; blue, 0 }  ,draw opacity=1 ][fill={rgb, 255:red, 183; green, 209; blue, 226 }  ,fill opacity=1 ] (429.8,344.57) -- (449.8,344.57) -- (449.8,375.9) -- (429.8,375.9) -- cycle ;
\draw  [color={rgb, 255:red, 0; green, 0; blue, 0 }  ,draw opacity=1 ][fill={rgb, 255:red, 183; green, 209; blue, 226 }  ,fill opacity=1 ] (472.83,314.57) -- (499.33,314.57) -- (499.33,345.07) -- (472.83,345.07) -- cycle ;
\draw  [color={rgb, 255:red, 0; green, 0; blue, 0 }  ,draw opacity=1 ][fill={rgb, 255:red, 183; green, 209; blue, 226 }  ,fill opacity=1 ] (515.5,290.57) -- (535.5,290.57) -- (535.5,320.57) -- (515.5,320.57) -- cycle ;
\draw  [color={rgb, 255:red, 0; green, 0; blue, 0 }  ,draw opacity=1 ][fill={rgb, 255:red, 183; green, 209; blue, 226 }  ,fill opacity=1 ] (481.67,376.5) -- (511.67,376.5) -- (511.67,396.5) -- (481.67,396.5) -- cycle ;
\draw  [color={rgb, 255:red, 49; green, 116; blue, 178 }  ,draw opacity=1 ][fill={rgb, 255:red, 49; green, 116; blue, 178 }  ,fill opacity=1 ] (488.67,348.07) .. controls (488.67,346.69) and (489.79,345.57) .. (491.17,345.57) .. controls (492.55,345.57) and (493.67,346.69) .. (493.67,348.07) .. controls (493.67,349.45) and (492.55,350.57) .. (491.17,350.57) .. controls (489.79,350.57) and (488.67,349.45) .. (488.67,348.07) -- cycle ;
\draw  [color={rgb, 255:red, 49; green, 116; blue, 178 }  ,draw opacity=0.3 ][line width=1.5]  (513.8,301.8) -- (513.8,373.07) -- (503.67,373.07) ;
\draw  [color={rgb, 255:red, 125; green, 172; blue, 214 }  ,draw opacity=0.2 ][line width=1.5]  (451.8,356.62) -- (451.8,301.4) -- (509.9,301.4) ;
\draw  [color={rgb, 255:red, 125; green, 172; blue, 214 }  ,draw opacity=0.9 ][line width=1.5]  (499.8,374.6) -- (453.17,374.6) -- (453.17,360.57) ;
\draw  [color={rgb, 255:red, 49; green, 116; blue, 178 }  ,draw opacity=0.9 ][line width=1.5]  (501.17,371.64) -- (491.17,371.64) -- (491.17,350.57) ;
\draw  [color={rgb, 255:red, 49; green, 116; blue, 178 }  ,draw opacity=1 ][line width=1.5]  (454.71,357.2) -- (454.71,348.07) -- (488.57,348.07) ;
\draw  [color={rgb, 255:red, 49; green, 116; blue, 178 }  ,draw opacity=0.9 ][line width=1.5]  (510.8,302.6) -- (510.8,348.07) -- (493.67,348.07) ;
\draw  [color={rgb, 255:red, 49; green, 116; blue, 178 }  ,draw opacity=1 ][fill={rgb, 255:red, 49; green, 116; blue, 178 }  ,fill opacity=1 ] (450.67,358.07) .. controls (450.67,356.69) and (451.79,355.57) .. (453.17,355.57) .. controls (454.55,355.57) and (455.67,356.69) .. (455.67,358.07) .. controls (455.67,359.45) and (454.55,360.57) .. (453.17,360.57) .. controls (451.79,360.57) and (450.67,359.45) .. (450.67,358.07) -- cycle ;
\draw  [color={rgb, 255:red, 49; green, 116; blue, 178 }  ,draw opacity=1 ][fill={rgb, 255:red, 49; green, 116; blue, 178 }  ,fill opacity=1 ] (498.67,373.07) .. controls (498.67,371.69) and (499.79,370.57) .. (501.17,370.57) .. controls (502.55,370.57) and (503.67,371.69) .. (503.67,373.07) .. controls (503.67,374.45) and (502.55,375.57) .. (501.17,375.57) .. controls (499.79,375.57) and (498.67,374.45) .. (498.67,373.07) -- cycle ;
\draw  [color={rgb, 255:red, 49; green, 116; blue, 178 }  ,draw opacity=1 ][fill={rgb, 255:red, 49; green, 116; blue, 178 }  ,fill opacity=1 ] (509.67,301.4) .. controls (509.67,300.02) and (510.79,298.9) .. (512.17,298.9) .. controls (513.55,298.9) and (514.67,300.02) .. (514.67,301.4) .. controls (514.67,302.78) and (513.55,303.9) .. (512.17,303.9) .. controls (510.79,303.9) and (509.67,302.78) .. (509.67,301.4) -- cycle ;

\draw  [fill={rgb, 255:red, 19; green, 97; blue, 157 }  ,fill opacity=1 ] (215,300) -- (365,300) .. controls (373.28,300) and (380,306.72) .. (380,315) -- (380,340) -- (200,340) -- (200,315) .. controls (200,306.72) and (206.72,300) .. (215,300) -- cycle ;
\draw  [fill={rgb, 255:red, 255; green, 255; blue, 255 }  ,fill opacity=1 ] (380,385) .. controls (380,393.28) and (373.28,400) .. (365,400) -- (215,400) .. controls (206.72,400) and (200,393.28) .. (200,385) -- (200,340) -- (380,340) -- (380,385) -- cycle ;
\draw  [fill={rgb, 255:red, 19; green, 97; blue, 157 }  ,fill opacity=1 ] (26,353) -- (176,353) .. controls (184.28,353) and (191,359.72) .. (191,368) -- (191,379) -- (11,379) -- (11,368) .. controls (11,359.72) and (17.72,353) .. (26,353) -- cycle ;
\draw  [fill={rgb, 255:red, 255; green, 255; blue, 255 }  ,fill opacity=1 ] (191,488) .. controls (191,496.28) and (184.28,503) .. (176,503) -- (26,503) .. controls (17.72,503) and (11,496.28) .. (11,488) -- (11,379) -- (191,379) -- (191,488) -- cycle ;
\draw  [color={rgb, 255:red, 0; green, 0; blue, 0 }  ,draw opacity=1 ][fill={rgb, 255:red, 183; green, 209; blue, 226 }  ,fill opacity=1 ] (46.33,442.07) -- (66.33,442.07) -- (66.33,473.4) -- (46.33,473.4) -- cycle ;
\draw  [color={rgb, 255:red, 0; green, 0; blue, 0 }  ,draw opacity=1 ][fill={rgb, 255:red, 183; green, 209; blue, 226 }  ,fill opacity=1 ] (88.83,412.07) -- (115.33,412.07) -- (115.33,442.57) -- (88.83,442.57) -- cycle ;
\draw  [color={rgb, 255:red, 0; green, 0; blue, 0 }  ,draw opacity=1 ][fill={rgb, 255:red, 183; green, 209; blue, 226 }  ,fill opacity=1 ] (131.33,388.07) -- (151.33,388.07) -- (151.33,418.07) -- (131.33,418.07) -- cycle ;
\draw  [color={rgb, 255:red, 0; green, 0; blue, 0 }  ,draw opacity=1 ][fill={rgb, 255:red, 183; green, 209; blue, 226 }  ,fill opacity=1 ] (97.67,473.4) -- (127.67,473.4) -- (127.67,493.4) -- (97.67,493.4) -- cycle ;
\draw  [color={rgb, 255:red, 49; green, 116; blue, 178 }  ,draw opacity=1 ][fill={rgb, 255:red, 49; green, 116; blue, 178 }  ,fill opacity=1 ] (125.67,398.9) .. controls (125.67,397.52) and (126.79,396.4) .. (128.17,396.4) .. controls (129.55,396.4) and (130.67,397.52) .. (130.67,398.9) .. controls (130.67,400.28) and (129.55,401.4) .. (128.17,401.4) .. controls (126.79,401.4) and (125.67,400.28) .. (125.67,398.9) -- cycle ;
\draw  [color={rgb, 255:red, 49; green, 116; blue, 178 }  ,draw opacity=1 ][fill={rgb, 255:red, 49; green, 116; blue, 178 }  ,fill opacity=1 ] (104.67,445.57) .. controls (104.67,444.19) and (105.79,443.07) .. (107.17,443.07) .. controls (108.55,443.07) and (109.67,444.19) .. (109.67,445.57) .. controls (109.67,446.95) and (108.55,448.07) .. (107.17,448.07) .. controls (105.79,448.07) and (104.67,446.95) .. (104.67,445.57) -- cycle ;
\draw  [color={rgb, 255:red, 49; green, 116; blue, 178 }  ,draw opacity=1 ][fill={rgb, 255:red, 49; green, 116; blue, 178 }  ,fill opacity=1 ] (66.67,455.57) .. controls (66.67,454.19) and (67.79,453.07) .. (69.17,453.07) .. controls (70.55,453.07) and (71.67,454.19) .. (71.67,455.57) .. controls (71.67,456.95) and (70.55,458.07) .. (69.17,458.07) .. controls (67.79,458.07) and (66.67,456.95) .. (66.67,455.57) -- cycle ;
\draw  [color={rgb, 255:red, 49; green, 116; blue, 178 }  ,draw opacity=1 ][fill={rgb, 255:red, 49; green, 116; blue, 178 }  ,fill opacity=1 ] (114.67,470.57) .. controls (114.67,469.19) and (115.79,468.07) .. (117.17,468.07) .. controls (118.55,468.07) and (119.67,469.19) .. (119.67,470.57) .. controls (119.67,471.95) and (118.55,473.07) .. (117.17,473.07) .. controls (115.79,473.07) and (114.67,471.95) .. (114.67,470.57) -- cycle ;
\draw [color={rgb, 255:red, 49; green, 116; blue, 178 }  ,draw opacity=1 ][line width=1.5]    (128.17,398.9) -- (122.35,410.9) ;
\draw [shift={(120.6,414.5)}, rotate = 295.88] [fill={rgb, 255:red, 49; green, 116; blue, 178 }  ,fill opacity=1 ][line width=0.08]  [draw opacity=0] (6.43,-3.09) -- (0,0) -- (6.43,3.09) -- (4.27,0) -- cycle    ;
\draw [color={rgb, 255:red, 49; green, 116; blue, 178 }  ,draw opacity=1 ][line width=1.5]    (69.17,455.57) -- (84.21,450.11) ;
\draw [shift={(87.97,448.75)}, rotate = 160.07] [fill={rgb, 255:red, 49; green, 116; blue, 178 }  ,fill opacity=1 ][line width=0.08]  [draw opacity=0] (6.43,-3.09) -- (0,0) -- (6.43,3.09) -- (4.27,0) -- cycle    ;
\draw [color={rgb, 255:red, 49; green, 116; blue, 178 }  ,draw opacity=1 ][line width=1.5]    (107.63,445.72) -- (110.29,452.99) ;
\draw [shift={(111.67,456.75)}, rotate = 249.92] [fill={rgb, 255:red, 49; green, 116; blue, 178 }  ,fill opacity=1 ][line width=0.08]  [draw opacity=0] (6.43,-3.09) -- (0,0) -- (6.43,3.09) -- (4.27,0) -- cycle    ;
\draw [color={rgb, 255:red, 49; green, 116; blue, 178 }  ,draw opacity=1 ][line width=1.5]    (117.17,470.57) -- (114.01,462.79) ;
\draw [shift={(112.5,459.08)}, rotate = 67.88] [fill={rgb, 255:red, 49; green, 116; blue, 178 }  ,fill opacity=1 ][line width=0.08]  [draw opacity=0] (6.43,-3.09) -- (0,0) -- (6.43,3.09) -- (4.27,0) -- cycle    ;

\draw  [fill={rgb, 255:red, 19; green, 97; blue, 157 }  ,fill opacity=1 ] (369,32) -- (479,32) .. controls (487.28,32) and (494,38.72) .. (494,47) -- (494,58) -- (354,58) -- (354,47) .. controls (354,38.72) and (360.72,32) .. (369,32) -- cycle ;
\draw  [fill={rgb, 255:red, 255; green, 255; blue, 255 }  ,fill opacity=1 ] (494,217) .. controls (494,225.28) and (487.28,232) .. (479,232) -- (369,232) .. controls (360.72,232) and (354,225.28) .. (354,217) -- (354,58) -- (494,58) -- (494,217) -- cycle ;

\draw [color={rgb, 255:red, 125; green, 172; blue, 214 }  ,draw opacity=1 ][line width=3]    (387.32,147.93) -- (438.37,132.96) ;
\draw [color={rgb, 255:red, 125; green, 172; blue, 214 }  ,draw opacity=1 ][line width=2.25]    (438.37,132.96) -- (463.02,86.81) ;
\draw [color={rgb, 255:red, 125; green, 172; blue, 214 }  ,draw opacity=1 ][line width=3.75]    (451.8,168.34) -- (438.37,132.96) ;
\draw [color={rgb, 255:red, 125; green, 172; blue, 214 }  ,draw opacity=1 ][line width=1.5]    (451.8,168.34) -- (463.02,86.81) ;
\draw [color={rgb, 255:red, 125; green, 172; blue, 214 }  ,draw opacity=1 ][line width=0.75]    (387.32,147.93) -- (463.02,86.81) ;
\draw [color={rgb, 255:red, 125; green, 172; blue, 214 }  ,draw opacity=1 ][line width=2.25]    (387.32,147.93) -- (451.8,168.34) ;
\draw  [color={rgb, 255:red, 49; green, 116; blue, 178 }  ,draw opacity=1 ][fill={rgb, 255:red, 49; green, 116; blue, 178 }  ,fill opacity=1 ] (456.3,86.81) .. controls (456.3,83.05) and (459.31,80.01) .. (463.02,80.01) .. controls (466.73,80.01) and (469.74,83.05) .. (469.74,86.81) .. controls (469.74,90.57) and (466.73,93.61) .. (463.02,93.61) .. controls (459.31,93.61) and (456.3,90.57) .. (456.3,86.81) -- cycle ;
\draw  [color={rgb, 255:red, 49; green, 116; blue, 178 }  ,draw opacity=1 ][fill={rgb, 255:red, 49; green, 116; blue, 178 }  ,fill opacity=1 ] (431.65,132.96) .. controls (431.65,129.21) and (434.66,126.16) .. (438.37,126.16) .. controls (442.07,126.16) and (445.08,129.21) .. (445.08,132.96) .. controls (445.08,136.72) and (442.07,139.77) .. (438.37,139.77) .. controls (434.66,139.77) and (431.65,136.72) .. (431.65,132.96) -- cycle ;
\draw  [color={rgb, 255:red, 49; green, 116; blue, 178 }  ,draw opacity=1 ][fill={rgb, 255:red, 49; green, 116; blue, 178 }  ,fill opacity=1 ] (380.61,147.93) .. controls (380.61,144.17) and (383.61,141.13) .. (387.32,141.13) .. controls (391.03,141.13) and (394.04,144.17) .. (394.04,147.93) .. controls (394.04,151.69) and (391.03,154.74) .. (387.32,154.74) .. controls (383.61,154.74) and (380.61,151.69) .. (380.61,147.93) -- cycle ;
\draw  [color={rgb, 255:red, 49; green, 116; blue, 178 }  ,draw opacity=1 ][fill={rgb, 255:red, 49; green, 116; blue, 178 }  ,fill opacity=1 ] (445.08,168.34) .. controls (445.08,164.59) and (448.09,161.54) .. (451.8,161.54) .. controls (455.51,161.54) and (458.51,164.59) .. (458.51,168.34) .. controls (458.51,172.1) and (455.51,175.15) .. (451.8,175.15) .. controls (448.09,175.15) and (445.08,172.1) .. (445.08,168.34) -- cycle ;

\draw [color={rgb, 255:red, 49; green, 116; blue, 178 }  ,draw opacity=1 ][line width=1.5]    (768,430.5) .. controls (757.66,449.3) and (747.32,459.26) .. (729.34,467.95) ;
\draw [shift={(725.8,469.6)}, rotate = 335.75] [fill={rgb, 255:red, 49; green, 116; blue, 178 }  ,fill opacity=1 ][line width=0.08]  [draw opacity=0] (11.61,-5.58) -- (0,0) -- (11.61,5.58) -- cycle    ;
\draw [color={rgb, 255:red, 49; green, 116; blue, 178 }  ,draw opacity=1 ][line width=1.5]    (767.67,373) .. controls (758.58,353.22) and (752.32,343.95) .. (727.06,332.27) ;
\draw [shift={(723.8,330.8)}, rotate = 23.89] [fill={rgb, 255:red, 49; green, 116; blue, 178 }  ,fill opacity=1 ][line width=0.08]  [draw opacity=0] (11.61,-5.58) -- (0,0) -- (11.61,5.58) -- cycle    ;
\draw [color={rgb, 255:red, 49; green, 116; blue, 178 }  ,draw opacity=1 ][line width=1.5]    (380.5,349.25) .. controls (396.82,372.29) and (399.32,379.2) .. (400.82,409.14) ;
\draw [shift={(401,413)}, rotate = 267.42] [fill={rgb, 255:red, 49; green, 116; blue, 178 }  ,fill opacity=1 ][line width=0.08]  [draw opacity=0] (11.61,-5.58) -- (0,0) -- (11.61,5.58) -- cycle    ;
\draw [color={rgb, 255:red, 49; green, 116; blue, 178 }  ,draw opacity=1 ][line width=1.5]    (337.67,197.33) .. controls (336.69,241.21) and (314.17,249.59) .. (321.39,295.86) ;
\draw [shift={(322,299.5)}, rotate = 259.7] [fill={rgb, 255:red, 49; green, 116; blue, 178 }  ,fill opacity=1 ][line width=0.08]  [draw opacity=0] (11.61,-5.58) -- (0,0) -- (11.61,5.58) -- cycle    ;
\draw  [fill={rgb, 255:red, 255; green, 255; blue, 255 }  ,fill opacity=1 ] (988,243.73) .. controls (988,251.89) and (981.81,258.5) .. (974.17,258.5) -- (835.83,258.5) .. controls (828.19,258.5) and (822,251.89) .. (822,243.73) -- (822,57.6) -- (988,57.6) -- (988,243.73) -- cycle ;
\draw  [fill={rgb, 255:red, 19; green, 97; blue, 157 }  ,fill opacity=1 ] (835.83,32) -- (974.17,32) .. controls (981.81,32) and (988,38.61) .. (988,46.77) -- (988,57.6) -- (822,57.6) -- (822,46.77) .. controls (822,38.61) and (828.19,32) .. (835.83,32) -- cycle ;
\draw  [color={rgb, 255:red, 19; green, 97; blue, 157 }  ,draw opacity=1 ][fill={rgb, 255:red, 47; green, 134; blue, 209 }  ,fill opacity=1 ] (867.67,90.5) -- (887.67,90.5) -- (887.67,110.5) -- (867.67,110.5) -- cycle ;
\draw  [color={rgb, 255:red, 19; green, 97; blue, 157 }  ,draw opacity=1 ][fill={rgb, 255:red, 140; green, 197; blue, 242 }  ,fill opacity=1 ] (887.67,90.5) -- (907.67,90.5) -- (907.67,110.5) -- (887.67,110.5) -- cycle ;
\draw  [color={rgb, 255:red, 19; green, 97; blue, 157 }  ,draw opacity=1 ][fill={rgb, 255:red, 220; green, 238; blue, 255 }  ,fill opacity=1 ] (907.67,90.5) -- (927.67,90.5) -- (927.67,110.5) -- (907.67,110.5) -- cycle ;
\draw  [color={rgb, 255:red, 19; green, 97; blue, 157 }  ,draw opacity=1 ][fill={rgb, 255:red, 140; green, 197; blue, 242 }  ,fill opacity=1 ] (867.67,110.5) -- (887.67,110.5) -- (887.67,130.5) -- (867.67,130.5) -- cycle ;
\draw  [color={rgb, 255:red, 19; green, 97; blue, 157 }  ,draw opacity=1 ][fill={rgb, 255:red, 15; green, 95; blue, 174 }  ,fill opacity=1 ] (887.67,110.5) -- (907.67,110.5) -- (907.67,130.5) -- (887.67,130.5) -- cycle ;
\draw  [color={rgb, 255:red, 19; green, 97; blue, 157 }  ,draw opacity=1 ][fill={rgb, 255:red, 185; green, 220; blue, 250 }  ,fill opacity=1 ] (907.67,110.5) -- (927.67,110.5) -- (927.67,130.5) -- (907.67,130.5) -- cycle ;
\draw  [color={rgb, 255:red, 19; green, 97; blue, 157 }  ,draw opacity=1 ][fill={rgb, 255:red, 220; green, 238; blue, 255 }  ,fill opacity=1 ] (867.67,130.5) -- (887.67,130.5) -- (887.67,150.5) -- (867.67,150.5) -- cycle ;
\draw  [color={rgb, 255:red, 19; green, 97; blue, 157 }  ,draw opacity=1 ][fill={rgb, 255:red, 185; green, 220; blue, 250 }  ,fill opacity=1 ] (887.67,130.5) -- (907.67,130.5) -- (907.67,150.5) -- (887.67,150.5) -- cycle ;
\draw  [color={rgb, 255:red, 19; green, 97; blue, 157 }  ,draw opacity=1 ][fill={rgb, 255:red, 47; green, 134; blue, 209 }  ,fill opacity=1 ] (907.67,130.5) -- (927.67,130.5) -- (927.67,150.5) -- (907.67,150.5) -- cycle ;

\draw  [color={rgb, 255:red, 19; green, 97; blue, 157 }  ,draw opacity=1 ][fill={rgb, 255:red, 255; green, 255; blue, 255 }  ,fill opacity=1 ] (927.67,90.5) -- (947.67,90.5) -- (947.67,110.5) -- (927.67,110.5) -- cycle ;
\draw  [color={rgb, 255:red, 19; green, 97; blue, 157 }  ,draw opacity=1 ][fill={rgb, 255:red, 255; green, 255; blue, 255 }  ,fill opacity=1 ] (927.67,110.5) -- (947.67,110.5) -- (947.67,130.5) -- (927.67,130.5) -- cycle ;
\draw  [color={rgb, 255:red, 19; green, 97; blue, 157 }  ,draw opacity=1 ][fill={rgb, 255:red, 255; green, 255; blue, 255 }  ,fill opacity=1 ] (927.67,130.5) -- (947.67,130.5) -- (947.67,150.5) -- (927.67,150.5) -- cycle ;
\draw  [color={rgb, 255:red, 19; green, 97; blue, 157 }  ,draw opacity=1 ][fill={rgb, 255:red, 255; green, 255; blue, 255 }  ,fill opacity=1 ] (907.67,150.5) -- (927.67,150.5) -- (927.67,170.5) -- (907.67,170.5) -- cycle ;
\draw  [color={rgb, 255:red, 19; green, 97; blue, 157 }  ,draw opacity=1 ][fill={rgb, 255:red, 255; green, 255; blue, 255 }  ,fill opacity=1 ] (887.67,150.5) -- (907.67,150.5) -- (907.67,170.5) -- (887.67,170.5) -- cycle ;
\draw  [color={rgb, 255:red, 19; green, 97; blue, 157 }  ,draw opacity=1 ][fill={rgb, 255:red, 255; green, 255; blue, 255 }  ,fill opacity=1 ] (867.67,150.5) -- (887.67,150.5) -- (887.67,170.5) -- (867.67,170.5) -- cycle ;
\draw  [color={rgb, 255:red, 19; green, 97; blue, 157 }  ,draw opacity=1 ][fill={rgb, 255:red, 255; green, 255; blue, 255 }  ,fill opacity=1 ] (927.67,150.5) -- (947.67,150.5) -- (947.67,170.5) -- (927.67,170.5) -- cycle ;
\draw  [color={rgb, 255:red, 19; green, 97; blue, 157 }  ,draw opacity=1 ][line width=1.5]  (867.67,90.5) -- (947.67,90.5) -- (947.67,170.5) -- (867.67,170.5) -- cycle ;

\draw  [fill={rgb, 255:red, 19; green, 97; blue, 157 }  ,fill opacity=1 ] (195,32) -- (305,32) .. controls (313.28,32) and (320,38.72) .. (320,47) -- (320,58) -- (180,58) -- (180,47) .. controls (180,38.72) and (186.72,32) .. (195,32) -- cycle ;
\draw  [fill={rgb, 255:red, 255; green, 255; blue, 255 }  ,fill opacity=1 ] (320,217) .. controls (320,225.28) and (313.28,232) .. (305,232) -- (195,232) .. controls (186.72,232) and (180,225.28) .. (180,217) -- (180,58) -- (320,58) -- (320,217) -- cycle ;

\draw  [color={rgb, 255:red, 0; green, 0; blue, 0 }  ,draw opacity=1 ][fill={rgb, 255:red, 183; green, 209; blue, 226 }  ,fill opacity=1 ] (198.47,135.73) -- (218.47,135.73) -- (218.47,167.07) -- (198.47,167.07) -- cycle ;
\draw  [color={rgb, 255:red, 0; green, 0; blue, 0 }  ,draw opacity=1 ][fill={rgb, 255:red, 183; green, 209; blue, 226 }  ,fill opacity=1 ] (241.5,105.73) -- (268,105.73) -- (268,136.23) -- (241.5,136.23) -- cycle ;
\draw  [color={rgb, 255:red, 0; green, 0; blue, 0 }  ,draw opacity=1 ][fill={rgb, 255:red, 183; green, 209; blue, 226 }  ,fill opacity=1 ] (284,81.73) -- (304,81.73) -- (304,111.73) -- (284,111.73) -- cycle ;
\draw  [color={rgb, 255:red, 0; green, 0; blue, 0 }  ,draw opacity=1 ][fill={rgb, 255:red, 183; green, 209; blue, 226 }  ,fill opacity=1 ] (250.33,167.67) -- (280.33,167.67) -- (280.33,187.67) -- (250.33,187.67) -- cycle ;
\draw  [color={rgb, 255:red, 125; green, 172; blue, 214 }  ,draw opacity=1 ][line width=1.5]  (282.33,93.92) -- (282.33,164.23) -- (272.33,164.23) ;
\draw  [color={rgb, 255:red, 125; green, 172; blue, 214 }  ,draw opacity=1 ][line width=1.5]  (220.47,147.57) -- (220.47,93.05) -- (278.59,93.05) ;
\draw  [color={rgb, 255:red, 125; green, 172; blue, 214 }  ,draw opacity=1 ][line width=1.5]  (267.83,165.52) -- (221.83,165.52) -- (221.83,151.73) ;
\draw  [color={rgb, 255:red, 125; green, 172; blue, 214 }  ,draw opacity=1 ][line width=1.5]  (269.83,162.67) -- (259.83,162.67) -- (259.83,141.73) ;
\draw  [color={rgb, 255:red, 125; green, 172; blue, 214 }  ,draw opacity=1 ][line width=1.5]  (223.38,147.57) -- (223.38,139.23) -- (257.5,139.23) ;
\draw  [color={rgb, 255:red, 125; green, 172; blue, 214 }  ,draw opacity=1 ][line width=1.5]  (279.5,93.75) -- (279.5,139.23) -- (262.33,139.23) ;
\draw  [color={rgb, 255:red, 49; green, 116; blue, 178 }  ,draw opacity=1 ][fill={rgb, 255:red, 49; green, 116; blue, 178 }  ,fill opacity=1 ] (267.33,164.23) .. controls (267.33,162.85) and (268.45,161.73) .. (269.83,161.73) .. controls (271.21,161.73) and (272.33,162.85) .. (272.33,164.23) .. controls (272.33,165.61) and (271.21,166.73) .. (269.83,166.73) .. controls (268.45,166.73) and (267.33,165.61) .. (267.33,164.23) -- cycle ;
\draw  [color={rgb, 255:red, 49; green, 116; blue, 178 }  ,draw opacity=1 ][fill={rgb, 255:red, 49; green, 116; blue, 178 }  ,fill opacity=1 ] (278.33,92.57) .. controls (278.33,91.19) and (279.45,90.07) .. (280.83,90.07) .. controls (282.21,90.07) and (283.33,91.19) .. (283.33,92.57) .. controls (283.33,93.95) and (282.21,95.07) .. (280.83,95.07) .. controls (279.45,95.07) and (278.33,93.95) .. (278.33,92.57) -- cycle ;
\draw  [color={rgb, 255:red, 49; green, 116; blue, 178 }  ,draw opacity=1 ][fill={rgb, 255:red, 49; green, 116; blue, 178 }  ,fill opacity=1 ] (257.33,139.23) .. controls (257.33,137.85) and (258.45,136.73) .. (259.83,136.73) .. controls (261.21,136.73) and (262.33,137.85) .. (262.33,139.23) .. controls (262.33,140.61) and (261.21,141.73) .. (259.83,141.73) .. controls (258.45,141.73) and (257.33,140.61) .. (257.33,139.23) -- cycle ;
\draw  [color={rgb, 255:red, 49; green, 116; blue, 178 }  ,draw opacity=1 ][fill={rgb, 255:red, 49; green, 116; blue, 178 }  ,fill opacity=1 ] (219.33,149.23) .. controls (219.33,147.85) and (220.45,146.73) .. (221.83,146.73) .. controls (223.21,146.73) and (224.33,147.85) .. (224.33,149.23) .. controls (224.33,150.61) and (223.21,151.73) .. (221.83,151.73) .. controls (220.45,151.73) and (219.33,150.61) .. (219.33,149.23) -- cycle ;

\node[forward group] at (165,10.5) {Graph Construction};
\node[forward group] at (672.5,10.5)
    {Algebraic Evaluation via the Matrix-Tree Theorem};
\node[forward group,font=\fontsize{12}{14}\selectfont\bfseries,
    anchor=east,inner xsep=2pt] at (318,279.57)
    {Fused Forward\&Backward};
\node[font=\fontsize{12}{14}\selectfont\bfseries,align=center,
    fill={rgb,255:red,219;green,235;blue,253},rounded corners=7.5pt,
    minimum width=59.25pt,minimum height=30pt,inner sep=2pt]
    at (59.27,291.83) {Gradient\\Descent};

\node[forward heading] at (80,45) {Pin Coordinates};
\node[forward heading] at (250,45) {Edge Weights};
\node[forward heading] at (424,45) {Edge Affinities};
\node[forward heading] at (658,45)
    {Graph Laplacian and Its Minor};
\node[forward heading] at (905,44.8) {Zero-Padded Inverse};
\node[forward heading] at (878,315) {Edge Marginal Probabilities};
\node[forward heading] at (569,267.5) {RMST Wirelength Evaluation};
\node[forward heading] at (290,320)
    {Spatial Derivatives of\\Edge Weights};
\node[forward heading] at (475,428) {RMST Wirelength Gradients};
\node[forward heading] at (101,366) {Gradient Directions};

\node at (80,210) {$\{(x_i,y_i)\mid i\in V_e\}$};
\node at (250,210) {$\boldsymbol{W}_e=[w_{ij}]$};
\node at (424,204)
    {$\boldsymbol{A}_e\!=\!\hat{\exp}\!\left(-\frac{\boldsymbol{W}_e}{\beta}\right)\!-\!\boldsymbol{I}$};
\node at (592,189) {$\boldsymbol{L}_e=\boldsymbol{D}_e-\boldsymbol{A}_e$};
\node at (720,184) {$\boldsymbol{L}_e^k$};
\node[forward note] at (717,201) {(remove $k$-th row\&col)};
\node at (905,219) {$\boldsymbol{M}_e=\begin{bmatrix}
    \left(\boldsymbol{L}_e^k\right)^{-1} & \boldsymbol{0}\\
    \boldsymbol{0}^{\mathsf T} & 0
\end{bmatrix}$};
\node at (878,472)
    {$\begin{aligned}
        \boldsymbol{P}_e&=\boldsymbol{A}_e\odot
            \bigl[\operatorname{Diag}(\boldsymbol{M}_e)\boldsymbol{1}^{\mathsf T}\\
        &\quad+\boldsymbol{1}\operatorname{Diag}(\boldsymbol{M}_e)^{\mathsf T}
            -2\boldsymbol{M}_e\bigr]
    \end{aligned}$};
\node at (628,346)
    {$\begin{aligned}
        \mathbb{E}_e[\mathcal{W}_{\text{RMST}}]
            &=\tfrac12\langle\boldsymbol{W}_e,\boldsymbol{P}_e\rangle\\
        &=\sum_{\substack{i,j\in V_e\\i<j}}w_{ij}p_{ij}
    \end{aligned}$};
\node at (290,371)
    {$\displaystyle\frac{\partial w_{ij}}{\partial x_i}
        \qquad\frac{\partial w_{ij}}{\partial y_i}$};
\node at (350,474)
    {$\displaystyle\frac{\partial\hat{\mathcal{W}}_{\text{RMST}}^\beta}{\partial x_i}
        =\sum\nolimits_{j\neq i}p_{ij}\frac{\partial w_{ij}}{\partial x_i}$};
\node at (600,474)
    {$\displaystyle\frac{\partial\hat{\mathcal{W}}_{\text{RMST}}^\beta}{\partial y_i}
        =\sum\nolimits_{j\neq i}p_{ij}\frac{\partial w_{ij}}{\partial y_i}$};

\end{tikzpicture}
}
\caption{Overview of the proposed DiffRMST wirelength model.}
\label{fig:forward}
\end{figure*}
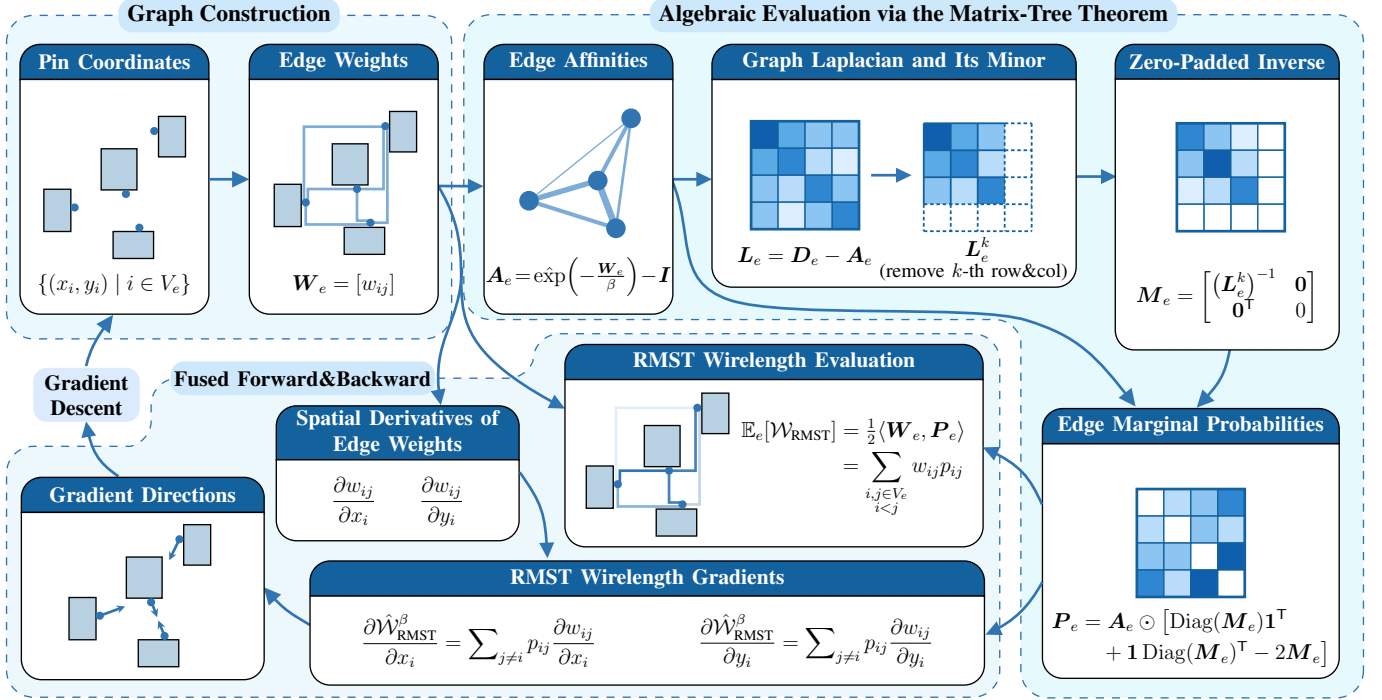

%% file: include/4-experiment.tex
\input{include/Table/exp_benchmarks}
\section{Experimental Results}\label{sec:exp}

This section evaluates the practical benefits of DiffRMST in terms of wirelength quality and computational efficiency. We compare its legalized placements and global routing results with those of the LSE and WA baselines to assess improvements in tree and routed wirelength. We then analyze wirelength changes across net degrees to examine how the benefits vary with net connectivity, and compare CPU and GPU execution to evaluate placement acceleration and complete placement time. DiffRMST is implemented within DREAMPlace~\cite{lin2021dreamplace}\footnote{\url{https://github.com/limbo018/DREAMPlace}} using C++/OpenMP for the CPU operator and CUDA for the GPU operator, with both integrated into the PyTorch placement framework. All experiments run on a server equipped with an AMD EPYC 7763 processor and an NVIDIA GeForce RTX 4090 GPU. The implementation is compiled in Release mode, and placement uses FP32. CPU operators use eight threads.

\subsection{Experimental Setup}
\label{sec:exp_setup}
\subsubsection{Benchmarks}
We evaluate DiffRMST on the ISPD2019 benchmark suite~\cite{liu2019ispd}, which includes designs with 8,208 to 899,325 movable instances and up to approximately 3.96 million pins. Table~\ref{tab:exp_benchmarks} summarizes the movable and fixed instance counts, net and pin counts, and net-degree distributions of the benchmark netlists.

Table~\ref{tab:exp_benchmarks} shows that two-pin nets account for 49.76\% to 80.37\% of all nets across the benchmark netlists. The proportion of larger nets varies substantially: nets with 33--99 pins account for approximately 5.90\% of the nets in \texttt{ispd\_test1} but less than 0.01\% of those in \texttt{ispd\_test4}. These differences allow us to evaluate the model and its parallel implementation across circuits with different sizes and connectivity patterns.

\subsubsection{Baselines}
We compare DiffRMST with DREAMPlace's native LSE~\cite{naylor2001non} and WA~\cite{hsu2013tsv} wirelength models.
The tables use the labels LSE-DREAMPlace, WA-DREAMPlace, and DiffRMST-DREAMPlace. All methods share the same benchmark configuration, initialization, and stopping criteria, with DREAMPlace's default settings for the remaining placement parameters. The baselines use their merged forward and backward operators, and DiffRMST uses the adaptive parameters in Section~\ref{sec:parallel_merged}. The main quality comparison evaluates the GPU implementations, while the efficiency comparison reports DiffRMST's CPU and GPU results and compares the GP speedups of all three models.

\subsubsection{Evaluation Metrics}
HPWL, RMST, and RSMT are evaluated on the final legalized DEF using the same nets with $2\leq|V_e|<100$ and unit net weights. We compute exact RMST using Manhattan Prim and exact RSMT using GeoSteiner~\cite{juhl2018geosteiner}. CUGR2.0~\cite{liu2023edge} evaluates routed wirelength (RWL) on the same layouts. RWL uses the router's access points and paths, whereas the geometric metrics use pin centers.

For runtime (RT), GP measures the time spent in global placement iterations, while Total includes input reading, initialization, global placement, legalization, and DEF output. Offline evaluation and routing are excluded. The RT column in Table~\ref{tab:exp_comparison} reports Total. GP speedup is measured relative to each model's CPU implementation. Table ratios and overall percentage improvements use DiffRMST as the reference, while individual examples and heatmap percentages use the compared baseline. Summary results are averaged across the benchmarks.

\subsection{Overall Comparison with Baseline Methods}
\label{sec:exp_quality}
Table~\ref{tab:exp_comparison} compares the three GPU implementations in terms of HPWL, RMST, RSMT, RWL, and complete placement runtime. Relative to LSE and WA, DiffRMST reduces RMST by 4.9\% and 5.0\% on average, respectively. The corresponding reductions in exact RSMT are 2.9\% and 3.2\%. Both tree metrics improve on nine of ten circuits against LSE and on all ten against WA. The agreement between RMST and RSMT shows that the benefit extends from the optimized tree objective to the independently evaluated Steiner-tree reference.

\input{include/Table/exp_comparison}

The improvement also carries over to global routing. DiffRMST reduces RWL by 3.2\% and 3.5\% on average relative to LSE and WA, respectively, with shorter RWL on eight and nine circuits. Considering the metrics jointly, DiffRMST achieves lower RMST, RSMT, and RWL than both baselines on eight of the ten benchmarks.

These gains persist on the four largest netlists, \texttt{ispd\_test7} through \texttt{ispd\_test10}, where all three metrics improve against both baselines. On \texttt{ispd\_test10}, RMST decreases by 6.09\% relative to LSE and 5.85\% relative to WA. The corresponding RSMT reductions are 4.11\% and 4.13\%, while RWL decreases by 4.31\% and 4.39\%.

HPWL exhibits the opposite trend: DiffRMST has larger HPWL on every benchmark against both baselines, with average increases of 4.9\% and 4.7\% relative to LSE and WA, respectively. This tradeoff is expected because DiffRMST optimizes a topology-aware tree objective rather than the bounding-box surrogate measured by HPWL, allowing higher HPWL while reducing tree-based and routed wirelength.

The wirelength improvements are achieved with acceptable runtime overhead. On \texttt{ispd\_test10}, DiffRMST completes placement in 45.81~s, compared with 40.74~s for LSE and 41.73~s for WA, while achieving the RMST, RSMT, and RWL reductions reported above.

\subsection{Wirelength Comparison Analysis across Net Degrees}
\label{sec:exp_degree}
Figure~\ref{fig:degree_analysis}(a) compares wirelength within individual net-degree groups. The upper curves show average baseline-to-DiffRMST wirelength ratios for the GPU implementations, with values above one favoring DiffRMST. The lower heatmap reports per-circuit RSMT changes within each group relative to WA, where negative values indicate shorter wirelength. All statistics cover the evaluated 2--99-pin nets. Each degree group includes all ten circuits except 65--99, which excludes the empty group in \texttt{ispd\_test4}.

\input{include/Figure/ispd2019_degree_comparison}

Both RMST and RSMT ratios exceed one against LSE and WA for every group with at least five pins, and the advantage becomes more pronounced above eight pins. For example, the WA-to-DiffRMST ratios for 17--32-pin nets reach 1.221 for RMST and 1.182 for RSMT. The heatmap confirms that RSMT decreases relative to WA on every circuit in the 9--16-pin and 17--32-pin groups, with more variation in the largest groups. In contrast, the average RSMT ratios for two-, three-, and four-pin nets fall below one against both baselines.

The heatmap also shows how the changes in smaller and larger nets combine. On \texttt{ispd\_test6} through \texttt{ispd\_test10}, RSMT increases in every 2--4-pin group and decreases in every group with at least five pins, while the All column remains negative. For these circuits, the savings in larger nets therefore outweigh the losses in smaller nets, despite the latter's predominance in net count in Table~\ref{tab:exp_benchmarks}.

The variation across benchmarks also relates to their net-degree distributions in Table~\ref{tab:exp_benchmarks}. In \texttt{ispd\_test1}, nets with 33--99 pins account for approximately 5.90\% of all nets in the input netlist. The heatmap shows RSMT reductions of 19.0\% and 22.2\% in the 33--64-pin and 65--99-pin groups, respectively, alongside a total reduction of 7.7\%. In comparison, \texttt{ispd\_test4} has a negligible proportion of nets in this degree range and a total reduction of 0.5\%. This contrast is consistent with greater scope for aggregate savings when high-degree nets are more prevalent, although their counts alone do not determine their contributions to total wirelength.

Longer low-degree nets in the final placement do not, by themselves, establish an intrinsic disadvantage of the proposed model on those nets. Equation~\eqref{eq:placement_obj} couples all net costs and density through shared cell coordinates. Moving a cell to shorten a larger net can lengthen its other incident nets, while density constraints limit compensating movements. The observed degree pattern is therefore consistent with improvements in larger nets consuming some of the placement freedom available to smaller nets. In particular, exact RMST, RSMT, and HPWL coincide for a two-pin net, and Equation~\eqref{eq:lse_obj} reduces to its single smoothed edge length without topological relaxation error.

The selected 26-pin net \texttt{net25438} in \texttt{ispd\_test2} illustrates the benefit of modeling internal connections. Figures~\ref{fig:net_placement_example}(b) and (c) compare the same net in the WA and DiffRMST placements. Colored cells and pins identify the selected net against gray neighboring cells. Although HPWL increases from 52.192 to 68.281~$\mu$m, RMST decreases from 131.120 to 92.320~$\mu$m and RSMT from 115.360 to 87.508~$\mu$m, giving tree-wirelength reductions of 29.59\% and 24.14\%, respectively. HPWL depends only on coordinate extrema and cannot distinguish internal pin arrangements with the same bounding box but different connection lengths. WA and LSE smooth this objective without explicitly modeling those connections. DiffRMST accounts for the connection structure and, in this example, produces shorter trees despite a larger bounding box, consistent with the gains in the higher-degree groups.

\subsection{Placement Acceleration}
\label{sec:exp_efficiency}
Table~\ref{tab:exp_efficiency} compares DiffRMST on the eight-thread CPU and the GPU. Both implementations complete legalization on all ten benchmarks, with final RMST differing by at most 1.69\% relative to the CPU result. GPU execution accelerates GP by $2.77\times$ to $49.48\times$, averaging $14.74\times$. The gains are particularly strong on \texttt{ispd\_test7} through \texttt{ispd\_test10}, where speedups reach $38.84\times$ to $49.48\times$ while RMST differences remain below 0.30\%. These circuits have similar degree distributions and large netlists, providing substantial concurrency across independent nets in addition to the parallel matrix operations within each multipin net.

\input{include/Table/exp_efficiency}

Figure~\ref{fig:gp_speedup} compares each model's GP speedup relative to its own CPU implementation. DiffRMST achieves higher speedups than both baselines on eight of the ten benchmarks. Its average speedup exceeds LSE's $8.68\times$ and WA's $9.07\times$. On \texttt{ispd\_test10}, LSE and WA achieve $26.25\times$ and $27.32\times$, compared with DiffRMST's $49.48\times$. This advantage is consistent with the degree-aware CUDA mapping in Section~\ref{sec:parallel_merged}, which distributes the additional matrix work of marginal evaluation across cooperating threads. These timings cover all GP iterations, including density computation.

\input{include/Figure/ispd2019_gp_speedup}

The average speedup for Total is $6.47\times$. On \texttt{ispd\_test10}, GPU global placement accounts for only 10.97~s of the 45.81~s Total runtime. The remaining runtime is dominated by initialization, legalization, output, and other non-GP stages, which reduces the impact of GP acceleration on the end-to-end runtime.

%% file: include/Table/exp_benchmarks.tex
\begin{table*}[t]
    \centering
    \caption{ISPD2019 benchmark statistics and net-degree distributions\textsuperscript{*}.}
    \label{tab:exp_benchmarks}
    \footnotesize
    \setlength{\tabcolsep}{2pt}
    \renewcommand{\arraystretch}{1.2}
    \begin{tabular*}{\textwidth}{@{\extracolsep{\fill}}c|rrrr*{10}{r}@{}}
\toprule
\multirow[c]{2}{*}{Benchmark} & \multicolumn{4}{c}{Netlist statistics} & \multicolumn{10}{c}{Net-degree distribution (\% of nets)} \\
\cline{2-5}\cline{6-15}
 & \#Movable & \#Fixed & \#Nets & \#Pins & 0--1 & 2 & 3 & 4 & 5--8 & 9--16 & 17--32 & 33--64 & 65--99 & $\geq 100$ \\
\midrule
\texttt{ispd\_test1} & 8,879 & 0 & 3,153 & 17,203 & 0.03 & 61.85 & 3.30 & 21.31 & 3.81 & 1.59 & 2.22 & 5.80 & 0.10 & 0.00 \\
\texttt{ispd\_test2} & 72,090 & 4 & 72,410 & 318,245 & 0.17 & 57.07 & 5.44 & 26.68 & 5.14 & 1.76 & 0.76 & 2.75 & 0.22 & $<0.01$ \\
\texttt{ispd\_test3} & 8,208 & 75 & 8,953 & 30,271 & 3.52 & 57.33 & 18.76 & 7.08 & 7.80 & 3.72 & 1.31 & 0.41 & 0.06 & 0.01 \\
\texttt{ispd\_test4} & 146,435 & 7 & 151,612 & 436,707 & 0.00 & 49.76 & 31.32 & 12.56 & 5.21 & 1.00 & 0.16 & $<0.01$ & 0.00 & $<0.01$ \\
\texttt{ispd\_test5} & 28,914 & 8 & 29,416 & 80,757 & 0.00 & 80.37 & 5.88 & 7.18 & 3.38 & 2.18 & 0.89 & 0.10 & 0.01 & 0.00 \\
\texttt{ispd\_test6} & 179,865 & 16 & 179,863 & 793,289 & 0.09 & 57.25 & 5.35 & 26.96 & 4.56 & 2.04 & 0.77 & 2.77 & 0.21 & $<0.01$ \\
\texttt{ispd\_test7} & 359,730 & 16 & 358,720 & 1,584,844 & 0.04 & 57.27 & 5.29 & 27.01 & 4.55 & 2.07 & 0.77 & 2.78 & 0.21 & $<0.01$ \\
\texttt{ispd\_test8} & 539,595 & 16 & 537,577 & 2,376,399 & 0.03 & 57.28 & 5.26 & 27.03 & 4.54 & 2.08 & 0.77 & 2.78 & 0.21 & $<0.01$ \\
\texttt{ispd\_test9} & 899,325 & 16 & 895,252 & 3,957,481 & 0.02 & 57.28 & 5.25 & 27.05 & 4.64 & 2.00 & 0.77 & 2.78 & 0.22 & $<0.01$ \\
\texttt{ispd\_test10} & 899,325 & 79 & 895,253 & 3,957,499 & 0.02 & 57.28 & 5.25 & 27.05 & 4.64 & 2.00 & 0.77 & 2.78 & 0.22 & $<0.01$ \\
\bottomrule
    \end{tabular*}
    \par\smallskip
    \begin{minipage}{\textwidth}
        \footnotesize \textsuperscript{*}Instance counts exclude fillers. Fixed instances exclude stand-alone I/O terminals. Net and pin counts describe the complete input netlist. Net degree is the number of pins per net. Percentages are relative to all nets in each input netlist. $<0.01$ denotes a nonzero percentage below 0.01\%.
    \end{minipage}
\end{table*}

%% file: include/Table/exp_comparison.tex
\begin{table*}[t]
    \centering
    \caption{Wirelength and complete placement runtime of the GPU implementations on ISPD2019. Best and second-best results are shown in \textbf{bold} and \textcolor{brown}{brown}, respectively\textsuperscript{*}.}
    \label{tab:exp_comparison}
    \footnotesize
    \setlength{\tabcolsep}{2pt}
    \renewcommand{\arraystretch}{1.2}
    \begin{tabular*}{\textwidth}{@{\extracolsep{\fill}}c|rrrrrrrrrrrrrrr@{}}
\toprule
\multirow[c]{2}{*}{Benchmark} & \multicolumn{5}{c}{LSE-DREAMPlace} & \multicolumn{5}{c}{WA-DREAMPlace} & \multicolumn{5}{c}{DiffRMST-DREAMPlace} \\
\cline{2-6}\cline{7-11}\cline{12-16}
 & HPWL & RMST & RSMT & RWL & RT & HPWL & RMST & RSMT & RWL & RT & HPWL & RMST & RSMT & RWL & RT \\
\midrule
\texttt{ispd\_test1} & \textbf{0.41} & \textcolor{brown}{0.63} & \textcolor{brown}{0.59} & \textcolor{brown}{0.58} & \textbf{2.83} & \textcolor{brown}{0.41} & 0.63 & 0.59 & 0.59 & \textcolor{brown}{3.06} & 0.47 & \textbf{0.57} & \textbf{0.55} & \textbf{0.54} & 3.19 \\
\texttt{ispd\_test2} & \textbf{17.68} & 22.38 & \textcolor{brown}{21.23} & \textcolor{brown}{21.00} & \textcolor{brown}{6.15} & \textcolor{brown}{17.69} & \textcolor{brown}{22.36} & 21.26 & 21.07 & \textbf{5.71} & 18.57 & \textbf{21.21} & \textbf{20.54} & \textbf{20.27} & 6.63 \\
\texttt{ispd\_test3} & \textbf{0.74} & \textcolor{brown}{0.90} & \textcolor{brown}{0.86} & \textcolor{brown}{0.84} & \textcolor{brown}{3.01} & \textcolor{brown}{0.76} & 0.92 & 0.87 & 0.85 & \textbf{2.99} & 0.79 & \textbf{0.89} & \textbf{0.85} & \textbf{0.83} & 3.08 \\
\texttt{ispd\_test4} & \textbf{17.17} & \textcolor{brown}{17.91} & \textcolor{brown}{17.63} & \textbf{24.37} & \textcolor{brown}{9.10} & \textcolor{brown}{17.22} & 17.96 & 17.70 & \textcolor{brown}{24.60} & \textbf{8.82} & 17.47 & \textbf{17.72} & \textbf{17.61} & 24.63 & 10.32 \\
\texttt{ispd\_test5} & \textbf{3.25} & \textbf{3.55} & \textbf{3.45} & \textbf{3.97} & \textbf{15.27} & \textcolor{brown}{3.34} & 3.67 & 3.57 & 4.06 & 15.76 & 3.42 & \textcolor{brown}{3.57} & \textcolor{brown}{3.52} & \textcolor{brown}{4.01} & \textcolor{brown}{15.50} \\
\texttt{ispd\_test6} & \textcolor{brown}{45.33} & 56.88 & 54.01 & 53.67 & \textbf{9.95} & \textbf{45.11} & \textcolor{brown}{56.62} & \textcolor{brown}{53.90} & \textcolor{brown}{53.58} & \textcolor{brown}{10.31} & 47.41 & \textbf{54.07} & \textbf{52.40} & \textbf{51.92} & 11.44 \\
\texttt{ispd\_test7} & \textcolor{brown}{92.48} & 116.25 & 110.30 & 109.53 & \textcolor{brown}{18.22} & \textbf{89.46} & \textcolor{brown}{112.90} & \textcolor{brown}{107.23} & \textcolor{brown}{106.42} & \textbf{17.04} & 95.22 & \textbf{108.97} & \textbf{105.41} & \textbf{104.36} & 19.05 \\
\texttt{ispd\_test8} & \textbf{134.75} & \textcolor{brown}{170.53} & \textcolor{brown}{161.58} & \textcolor{brown}{160.10} & \textcolor{brown}{27.24} & \textcolor{brown}{135.55} & 171.16 & 162.54 & 161.28 & \textbf{26.28} & 140.84 & \textbf{161.35} & \textbf{155.99} & \textbf{154.10} & 27.63 \\
\texttt{ispd\_test9} & \textcolor{brown}{213.26} & 273.11 & 257.40 & 253.95 & \textcolor{brown}{41.62} & \textbf{211.12} & \textcolor{brown}{269.28} & \textcolor{brown}{254.63} & \textcolor{brown}{251.41} & \textbf{39.29} & 220.72 & \textbf{256.44} & \textbf{246.60} & \textbf{242.65} & 44.93 \\
\texttt{ispd\_test10} & \textbf{213.06} & 272.35 & \textcolor{brown}{257.35} & \textcolor{brown}{254.91} & \textbf{40.74} & \textcolor{brown}{213.16} & \textcolor{brown}{271.64} & 257.40 & 255.11 & \textcolor{brown}{41.73} & 221.58 & \textbf{255.77} & \textbf{246.76} & \textbf{243.92} & 45.81 \\
\midrule
Ratio & \textbf{0.951} & \textcolor{brown}{1.049} & \textcolor{brown}{1.029} & \textcolor{brown}{1.032} & \textcolor{brown}{0.928} & \textcolor{brown}{0.953} & 1.050 & 1.032 & 1.035 & \textbf{0.919} & 1.000 & \textbf{1.000} & \textbf{1.000} & \textbf{1.000} & 1.000 \\
\bottomrule
    \end{tabular*}
    \par\smallskip
    \begin{minipage}{\textwidth}
        \footnotesize \textsuperscript{*}Wirelength ($10^6$~$\mu$m), RT (s). Ratio reports average ratios relative to DiffRMST-DREAMPlace. For each metric, lower values are better. Rankings use unrounded values.
    \end{minipage}
\end{table*}

%% file: include/Figure/ispd2019_degree_comparison.tex
\begin{figure*}[t]
    \centering
    \includegraphics[width=\textwidth]{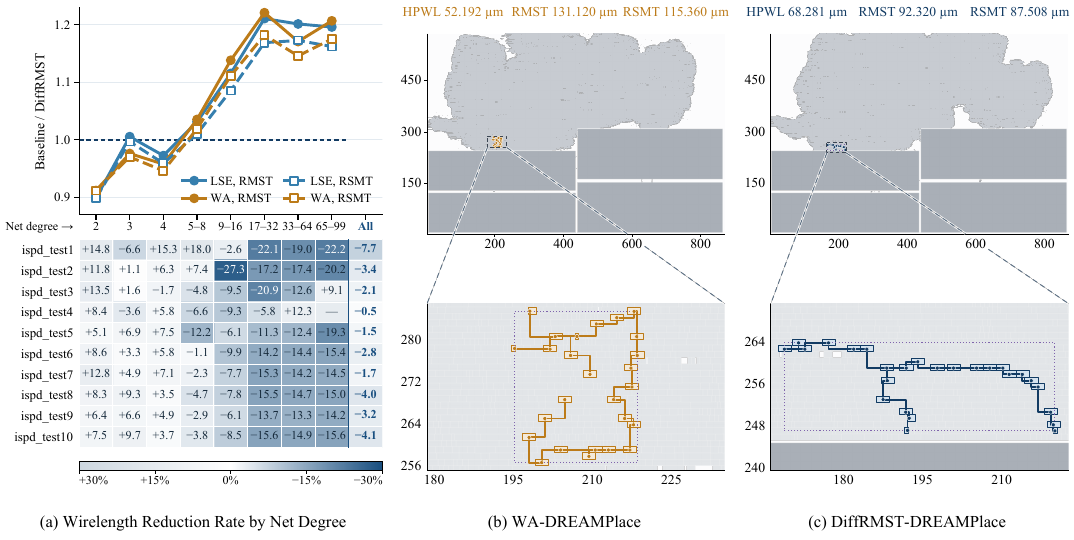}
    \caption{Wirelength comparison across net degrees. (a) Average baseline-to-DiffRMST ratios (top) and RSMT percentage changes relative to WA (bottom). Ratios above one and negative changes favor DiffRMST. All denotes each circuit's total over evaluated nets. (b) WA and (c) DiffRMST placements of the same 26-pin net in \texttt{ispd\_test2}, with RMST overlays at matched scales. Dotted boxes enclose pins, and dashed guides indicate enlarged regions. Placement coordinates and wirelengths are in $\mu$m.}
    \label{fig:degree_analysis}
    \label{fig:rsmt_heatmap}
    \label{fig:net_placement_example}
\end{figure*}

%% file: include/Table/exp_efficiency.tex
\begin{table}[t]
    \centering
    \caption{CPU and GPU comparison of DiffRMST-DREAMPlace. RMST ($10^6$), GP and Total (s). Ratio reports average ratios, with GPU normalized to 1. Best and second-best results are shown in \textbf{bold} and \textcolor{brown}{brown}, respectively\textsuperscript{*}.}
    \label{tab:exp_efficiency}
    \footnotesize
    \setlength{\tabcolsep}{2pt}
    \renewcommand{\arraystretch}{1.2}
    \begin{tabular*}{\columnwidth}{@{\extracolsep{\fill}}c|rrrrrr@{}}
\toprule
\multirow[c]{2}{*}{Benchmark} & \multicolumn{3}{c}{CPU} & \multicolumn{3}{c}{GPU} \\
\cline{2-4}\cline{5-7}
 & RMST & GP & Total & RMST & GP & Total \\
\midrule
\texttt{ispd\_test1} & \textbf{0.56} & \textcolor{brown}{6.10} & \textcolor{brown}{7.10} & \textcolor{brown}{0.57} & \textbf{2.20} & \textbf{3.19} \\
\texttt{ispd\_test2} & \textbf{21.20} & \textcolor{brown}{33.11} & \textcolor{brown}{37.78} & \textcolor{brown}{21.21} & \textbf{2.68} & \textbf{6.63} \\
\texttt{ispd\_test3} & \textcolor{brown}{0.90} & \textcolor{brown}{11.56} & \textcolor{brown}{12.63} & \textbf{0.89} & \textbf{2.09} & \textbf{3.08} \\
\texttt{ispd\_test4} & \textbf{17.71} & \textcolor{brown}{41.79} & \textcolor{brown}{48.15} & \textcolor{brown}{17.72} & \textbf{5.31} & \textbf{10.32} \\
\texttt{ispd\_test5} & \textcolor{brown}{3.57} & \textcolor{brown}{48.30} & \textcolor{brown}{52.00} & \textbf{3.57} & \textbf{13.56} & \textbf{15.50} \\
\texttt{ispd\_test6} & \textbf{54.06} & \textcolor{brown}{74.40} & \textcolor{brown}{86.46} & \textcolor{brown}{54.07} & \textbf{3.28} & \textbf{11.44} \\
\texttt{ispd\_test7} & \textbf{108.65} & \textcolor{brown}{198.84} & \textcolor{brown}{218.91} & \textcolor{brown}{108.97} & \textbf{5.12} & \textbf{19.05} \\
\texttt{ispd\_test8} & \textcolor{brown}{161.44} & \textcolor{brown}{286.28} & \textcolor{brown}{319.74} & \textbf{161.35} & \textbf{6.60} & \textbf{27.63} \\
\texttt{ispd\_test9} & \textcolor{brown}{256.59} & \textcolor{brown}{502.63} & \textcolor{brown}{560.22} & \textbf{256.44} & \textbf{10.42} & \textbf{44.93} \\
\texttt{ispd\_test10} & \textbf{255.01} & \textcolor{brown}{542.79} & \textcolor{brown}{579.86} & \textcolor{brown}{255.77} & \textbf{10.97} & \textbf{45.81} \\
\midrule
Ratio & \textbf{0.999} & \textcolor{brown}{14.741} & \textcolor{brown}{6.473} & \textcolor{brown}{1.000} & \textbf{1.000} & \textbf{1.000} \\
\bottomrule
    \end{tabular*}
    \par\smallskip
    \begin{minipage}{\columnwidth}
        \footnotesize \textsuperscript{*}For each metric, lower values are better. RMST values that coincide after rounding are ranked using the original measurements.
    \end{minipage}
\end{table}

%% file: include/Figure/ispd2019_gp_speedup.tex
\begin{figure}[t]
    \centering
    \includegraphics[width=\columnwidth]{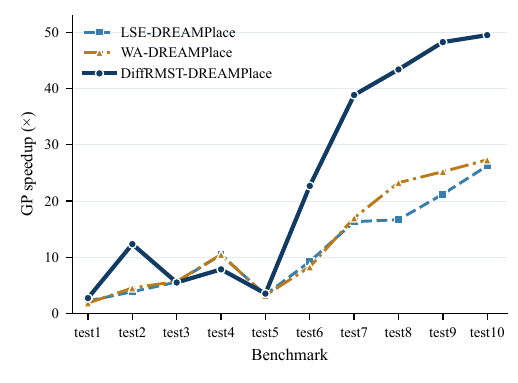}
    \caption{GPU speedup in global placement relative to each model's eight-thread CPU implementation.}
    \label{fig:gp_speedup}
\end{figure}

%% file: include/5-conclusion.tex
\section{Conclusion}\label{sec:conclusion}

This paper presented DiffRMST, a fully differentiable RMST wirelength model that combines spatial smoothing, LSE relaxation, and the Kirchhoff matrix-tree theorem to provide topology-aware gradients for analytical global placement. Approximation bounds and gradient regularity analysis characterize the relaxed objective and its relationship to exact RMST. A fused forward and backward formulation supports parallel CPU and GPU implementations. On the ISPD2019 benchmark suite, DiffRMST GPU reduces RMST, RSMT, and routed wirelength by 5.0\%, 3.2\%, and 3.5\% on average, respectively, relative to WA GPU. These gains are achieved with acceptable runtime overhead. The benefits are more pronounced for higher-degree nets, highlighting the value of modeling internal connectivity beyond the bounding box. Future work will explore extensions to timing-driven and congestion-aware placement, as well as applications to 3D placement.

\section*{Acknowledgments}
The authors disclose the use of generative AI tools for language editing and sentence refinement to enhance readability and clarity. All technical content, methodology, and results were produced and verified by the authors.